\documentclass[11pt]{article}
\usepackage[T1]{fontenc}
\usepackage[margin=1in]{geometry}
\usepackage{mathtools,amssymb,amsthm}
\usepackage{microtype}
\usepackage{graphicx,booktabs,enumitem}
\usepackage{xcolor}
\usepackage{tikz}
\usetikzlibrary{arrows.meta}

\usepackage[
  colorlinks=true,
  linkcolor=red,  
  citecolor=red,  
  urlcolor=red    
]{hyperref}

\usepackage[capitalize,nameinlink,noabbrev]{cleveref}
\usepackage{braket}

\usepackage{algorithm}
\usepackage{algpseudocode}
\algrenewcommand{\algorithmicrequire}{\textbf{Input:}}
\algrenewcommand{\algorithmicensure}{\textbf{Output:}}
\floatname{algorithm}{Protocol}

\theoremstyle{plain}
\newtheorem{theorem}{Theorem}[section]
\newtheorem{lemma}[theorem]{Lemma}
\newtheorem{corollary}[theorem]{Corollary}

\newtheorem{claim}[theorem]{Claim}

\theoremstyle{definition}
\newtheorem{definition}[theorem]{Definition}

\theoremstyle{remark}
\newtheorem{remark}[theorem]{Remark}

\numberwithin{equation}{section}

\newcommand{\R}{\mathbb{R}}
\newcommand{\C}{\mathbb{C}}

\newcommand{\E}{\mathbb{E}}

\DeclareMathOperator{\Var}{Var}

\newcommand{\ba}{\mathbf{a}}

\newcommand{\bx}{\mathbf{x}}
\newcommand{\by}{\mathbf{y}}
\newcommand{\bz}{\mathbf{z}}

\definecolor{HaoyuBlue}{RGB}{25,90,175}

\newcommand{\sugroup}{\operatorname{SU}}
\newcommand{\fun}[1]{\mathsf{#1}}
\newcommand{\trace}{\operatorname{Tr}}
\newcommand{\ugroup}{\operatorname{U}}

\title{Exponential Quantum Advantage in Number-on-Forehead Communication}
\author{%
  Haoyu Wang\thanks{Penn State University. Email: \href{mailto:hjw5492@psu.edu}{\texttt{hjw5492@psu.edu}}.}
  \and
  Pei Wu\thanks{Penn State University. Email: \href{mailto:pei.wu@psu.edu}{\texttt{pei.wu@psu.edu}}.}
  \and
  Guangxu Yang\thanks{University of Southern California. Email: \href{mailto:guangxuy@usc.edu}{\texttt{guangxuy@usc.edu}}.}%
}
\date{}

\begin{document}
\maketitle

\begin{abstract}
We give the first exponential quantum advantage in the general interactive three-party Number-on-Forehead (NOF) model for a decision problem. Previous separations hold only for restricted protocols like one-way communication for a relation. We construct an explicit partial Boolean function, the Interleaved Unitary Product problem, that requires only $O(\log n)$  quantum communication but  $\widetilde{\Omega}(n^{1/32})$ randomized communication in the NOF model. This function builds on the two-party unitary product problem of Arunachalam, Girish, and Lifshitz (TQC 2024).

The main technical obstacle is that discrepancy, the standard lower-bound method for NOF, also lower-bounds quantum communication. We instead develop a regularity-based argument for randomized NOF lower bounds, building on the approach of Kelley, Lovett, and Meka (STOC 2024) and adapting the regularity decomposition of Abboud, Fischer, Kelley, Lovett, and Meka (STOC 2024) to cylinder intersections.
Combined with matrix-product estimates of Arunachalam, Girish, and Lifshitz (TQC 2024), this yields our randomized lower bound.
\end{abstract}

\section{Introduction}\label{sec:intro}
When can quantum messages convey information more  efficiently than classical messages?
Communication complexity, introduced by Yao~\cite{Yao}, provides a natural setting in which this question can be studied through unconditional lower bounds. A long line of work has established exponential separations between quantum and classical communication complexity~\cite{buhrman1998quantum,raz1999exponential,bar2004exponential,gavinsky2007exponential,regev2011quantum,gavinsky2016entangled,girish2022quantum,gavinsky2019quantum,gavinsky2020bare,goos2024quantum,gavinsky2026total}. These separations identify tasks for which quantum messages convey useful information far more efficiently than classical messages, without relying on unproved computational hardness assumptions.

Communication separations have also motivated experimental demonstrations of quantum advantage~\cite{trojek2005experimental,xu2015experimental,guan2016observation,kumar2019experimental,shen2025experimental}.
One example is \emph{quantum information supremacy}, proposed by Aaronson, Buhrman, and Kretschmer~\cite{ABK24}, which uses one-way communication separations to demonstrate an advantage in the information resources required for a task. Recently, a trapped-ion experiment has demonstrated such an advantage, supported by unconditional classical lower bounds~\cite{KGDGGGHMNHA25}. These developments motivate understanding the scope of quantum communication advantage along two dimensions: the information available to the players and the quantum resources available to their protocols.

\paragraph{Quantum advantage with overlapping inputs.}
The separations discussed above establish quantum advantages in two-party models, where the players have access to disjoint parts of the input. Despite this progress, much less is understood about quantum communication advantage when the players' input access overlaps. The importance of input-access assumptions is also highlighted by the trapped-ion experiment~\cite{KGDGGGHMNHA25}: the authors identify a potential loophole arising from the fact that Bob's measurement choice was generated before Alice's state
was prepared, so the preparation device could in principle have advance knowledge of that choice. 
This motivates a natural research question:
\emph{How does overlap in the players’ input access affect quantum communication advantage?}


\paragraph{Quantum advantage with one clean qubit.}
A complementary question is: \emph{What are the minimal quantum resources needed for an exponential communication advantage?} We focus on initial-state purity. Motivated by nuclear magnetic resonance implementations, where the initial state can be highly mixed, Knill and Laflamme introduced the one-clean-qubit model of quantum computation, also known as DQC1~\cite{knill1998power}. Only one qubit is initially pure,
while the remaining workspace is maximally mixed. The model thus
restricts initial-state purity while allowing ideal unitary control. Klauck and Lim introduced a communication version of this model~\cite{klauck2018power}, making it possible to investigate the power of limited purity through unconditional communication lower bounds. Arunachalam, Girish, and Lifshitz subsequently established an exponential separation between one-clean-qubit quantum communication and interactive randomized communication in the two-party setting~\cite{arunachalam2023one}.

\medskip
\medskip
The Number-on-Forehead (NOF) model, introduced by Chandra, Furst, and Lipton~\cite{chandra1983multi}, is one of the best-known models of multiparty communication and captures a maximal form of input overlap. Lower bounds in the NOF model have applications in various areas including circuit complexity \cite{HG90,beigel1994acc,chakrabarti2007lower,VW07}, additive
combinatorics~\cite{linial2017communication,jaber2025}, distributed computing~\cite{drucker2014power}, pseudorandom generators~\cite{babai1992multiparty,ganor2014space},  streaming algorithms~\cite{kallaugher2019complexity}, and cryptography~\cite{corrigan2019function,miles2013shielding,li2020leakage}. 
It therefore provides a natural setting that brings these two questions together. The input consists of $k$ parts, $x_1,\ldots,x_k$, and player $i$ sees every part except $x_i$. Thus, each input part is visible to all but one player, while no player sees the entire input. For $k=2$, this recovers the standard two-party model, up to relabeling the inputs, while for $k\geq 3$, the players' views overlap. 


The above considerations lead to the following concrete question in multiparty communication:
\begin{quote}
\centering\itshape
Does one clean qubit suffice for an exponential quantum advantage in the multi-party NOF communication model?
\end{quote}

In this paper, we answer the above question affirmatively in the three-party NOF model.
We establish the first unconditional exponential separation between quantum and randomized communication for a partial Boolean function in the three-party NOF model. 

\subsection{Our contributions}
To obtain this separation, we introduce the Interleaved Unitary Product problem, which is  inspired by the two-party communication problem in~\cite{arunachalam2023one}.
\begin{definition}[Interleaved Unitary Product (IUP) problem]
\label{def:overview-abcde}
Let $m\ge1$ be an integer. The input consists of complex matrices $A_i,B_i,C_i,D_i,E_i\in\mathbb C^{m\times m}$,  $i\in\{1,2\}$. Each input matrix $U$ is promised to satisfy $\|U^\dagger U-I_m\|_{\mathrm F}\le10^{-4}$. Group the inputs as
\[
X=(A_1,C_1,A_2,C_2),\quad
Y=(B_1,D_1,B_2,D_2),\quad
Z=(E_1,E_2).
\]

Define $W=A_1B_1C_1D_1E_1A_2B_2C_2D_2E_2$. In the three-party NOF model, Alice, Bob, and Charlie see $(Y,Z)$, $(X,Z)$, and $(X,Y)$, respectively. Their task is to output
\[
\fun b=
\begin{cases}
1, & \text{if }\operatorname{Re}\operatorname{Tr}(W)/m\ge0.9,\\
0, & \text{if }\operatorname{Re}\operatorname{Tr}(W)/m\le0.1
\end{cases}
\]
under the promise that $\operatorname{Re}\operatorname{Tr}(W)/m$ doesn't fall into $(0.1, 0.9)$.
\end{definition}

To obtain a Boolean function, encode $X, Y,$ and $Z$ as bit strings $\bx,\by,\bz\in\{0,1\}^n$, using $O(\log n)$ bits for each matrix entry and padding $\bz$ with zeros. This allows matrices of dimension $m=\Theta(\sqrt{n/\log n})$; see Definition~\ref{def:pro} for the formal description.  

Our first result shows that the interleaved unitary product problem can be tested with logarithmic quantum communication under the one-clean-qubit guarantee. The protocol is a natural adaptation of that in~\cite{arunachalam2023one} and is presented with its analysis  in Section~\ref{sec:quantum_protocol}.  
\begin{theorem}[Quantum upper bound]
\label{thm:intro-quantum}
There is a bounded-error one-clean-qubit quantum NOF protocol for the IUP problem using $O(\log n)$ qubits of communication. \footnote{Here, ``one clean qubit'' means one per run with a constant acceptance-probability gap.
We use $O(1)$ independent repetitions to achieve any fixed constant error,
preserving $O(\log n)$ total communication. Such repetitions are allowed in the original DQC1 model of
Knill and Laflamme~\cite{knill1998power}.
Arunachalam, Girish, and Lifshitz~\cite{arunachalam2023one}
likewise use $O(1)$ repetitions for amplification, using 
$O(1)$ clean qubits in total. Klauck and Lim~\cite{klauck2018power} instead incorporate the amplification overhead into the cost.}
\end{theorem}

In contrast, every classical randomized protocol requires polynomial communication, even with unrestricted interaction among all three players. This lower bound is our main technical contribution. 
\begin{theorem}[Randomized lower bound]
\label{thm:intro-classical}
Every bounded-error randomized NOF protocol for the IUP problem requires $\widetilde{\Omega}\!\left(n^{1/32}\right)$
bits of communication.
\end{theorem}

We prove this lower bound in Section~\ref{sec:lowerbound}. Together, Theorems~\ref{thm:intro-quantum} and~\ref{thm:intro-classical} give the claimed exponential separation for a partial Boolean function, while requiring only one clean qubit on the quantum side.

\subsection{Related work}
Gavinsky and Pudl\'ak~\cite{QCCinNOF} gave relational problems with $O(\log n)$-qubit simultaneous quantum protocols and $n^{\Omega(1/k^2)}$ randomized lower bounds in a non-interactive NOF model. Their separation is exponential for constant $k$ and superpolynomial for $k=o\bigl(\sqrt{\log n/\log\log n}\bigr)$. In their classical model, the first $k-1$ players each send a message to the final player without seeing one another's messages. Yang and Zhang~\cite{yang2025quantum} gave a gadgeted Hidden Matching construction for three-party simultaneous NOF, obtaining an $O(\log n)$ quantum upper bound and an $\Omega(n^{1/16})$ randomized lower bound. Their subsequent work~\cite{yang2026exponential} extended this direction to one-way NOF communication, giving an $O(\log n)$ quantum upper bound and an $\Omega(n^{1/3}/2^{k/3})$ randomized lower bound. These results consider relational problems and restrictions on classical interaction. Our result allows \emph{unrestricted} interaction among all three classical players and establishes the separation for a partial \emph{Boolean} function.

\subsection{Technical Overview}
\label{subsec:technical-overviews}
The bulk of technical work in this paper is the classical lower bound in the NOF model. 
The overlap of information among the players is exactly what makes NOF lower bounds notoriously difficult to prove. 
A quantum-classical separation faces a second obstacle: the discrepancy and generalized discrepancy methods, which underlie many known lower bounds for interactive NOF protocols, also lower bound quantum communication, and therefore cannot separate the two models.

Our classical lower bound builds on a strong corruption bound of cylinder intersections, which shows 
that distinguishing two well-separated  distributions requires the claimed
communication cost. 
To prove this corruption bound, we adapt the regularity decomposition for Boolean matrices developed by Abboud, Fischer, Kelley, Lovett, and Meka~\cite{abboud2024new} to cylinder intersections, which allows us to treat the structure (sparse here) and pseudorandom cases separately.
Besides, the matrix-product estimates of Arunachalam, Girish, and Lifshitz~\cite{arunachalam2023one} and the regularity estimates of Kelley, Lovett, and Meka~\cite{kelley2024explicit} provide the key technical tools for establishing the corruption bound.

\paragraph{Hard distributions and corruption bound.}
We first work with exact matrices in $\mathrm{SU}(m)$.
Under $\mu_0$, all inputs are independent Haar matrices; under \(\mu_1\), all but the last matrix are sampled independently from the Haar measure. The last matrix is then set to the inverse of their product, ensuring that \(W=I_m\).
Thus $\mu_0$ produces $0$-inputs with high probability while $\mu_1$ produces
$1$-inputs. On the other hand, Haar invariance makes every player's view identically
distributed under both measures: the distinguishing information lies in the
global correlation. In the NOF model, each fixed transcript corresponds to a cylinder
intersection with indicator
$\fun c(X,Y,Z)=\fun f(Y,Z)\fun g(X,Z)\fun h(X,Y)$.
We prove $\mu_1(\fun c)\le\alpha\mu_0(\fun c)+e^{-\beta m^{1/16}}$
for constants $\alpha>1$ and $\beta>0$.

\paragraph{Step 1. Decompose cylinder intersections.}
Fix $d = \Theta(m^{1/16})$.
By our cylinder intersection  decomposition adapted from
~\cite[Theorem~3.1]{abboud2024new},  we decompose $\fun c$ into $k$ disjoint pieces:  
$\mathsf{c}_i=\fun f_i \fun g_i \fun h_i$ on $\mathcal B_i:= \mathcal{X}_{i}\times \mathcal{Y}_{i}\times \mathcal{Z}_{i}$,  $1\le i \le k$, where  
$
\fun c=\sum_{i=1}^k \fun c_i$, $
k\le e^{O(d^7)}$ and $
\sum_i\mu_0(\mathcal B_i)=O(d^2).
$ For $\fun u_i\in\{\fun f_i,\fun g_i,\fun h_i\}$,
let $\delta_{\fun u,i}:=\E_{\mathcal B_i}[\fun u_i]$. There are two kinds of pieces: one kind of pieces is \textit{sparse},  satisfying  $\delta_{f,i}\le2^{-d}$, $\delta_{g,i}\le2^{-d}$,
or $\delta_{h,i}\le2^{-\Omega(d)}$;
the remaining pieces are \textit{pseudorandom}: both $\fun f_i$ and $\fun g_i$
satisfy the regularity and min-degree properties of
\cite{kelley2024explicit}, which  allow their product average to be compared
with the product of their averages.

\paragraph {Step 2. Derive the grid-norm bound.}
The regularity notion of \cite{kelley2024explicit} uses grid norms $\|\cdot \|_{U(\ell,d)}$ to measure uniformity relative to density.
From the matrix-product estimate of \cite{arunachalam2023one}, we derive the following grid norm bound: for every integer $p\ge2$,
\begin{equation}\label{eq:overview-grid-bound}
\mu_1(\fun c_i)\le
O\!\left(\mu_0 (\mathcal B_i)\cdot \delta_{h,i}^{1-1/p}
\cdot \|\fun f_i\|_{U(2,p),i}\cdot \|\fun g_i\|_{U(2,p),i}\right) +e^{-\Omega(\sqrt m/p)},
\end{equation}
where the norms $\|\cdot\|_{U{(2,p)},i}$ are local versions on $\mathcal B_i$.

\paragraph{Step 3. Establish the corruption bound.}
For regular pieces, the regularity property bounds the two norms by the local expectations, i.e., $\|\fun f_i\|_{U(2,d),i} \|\fun g_i\|_{U(2,d),i}  = O(\delta_{f,i}\delta_{g,i})$; besides, 
the estimate of \cite{kelley2024explicit} also gives
$\mu_0(\fun c_i)\ge\Omega(\mu_0(\mathcal B_i)\delta_{f,i}\delta_{g,i}\delta_{h,i})$.
Together with~\eqref{eq:overview-grid-bound}, this yields the corruption bound on  regular cylinder intersections:    
$\mu_1(\fun c_i)\le O(\mu_0(\fun c_i))+e^{-\Omega(\sqrt m/d)}$.

For sparse pieces, 
by  
standard calculation  $\|\fun f_i\|_{U(2,2),i} \|\fun g_i\|_{U(2,2),i}  = O(\sqrt{\delta_{f,i}\delta_{g,i}})$;  applying~\eqref{eq:overview-grid-bound} with $p=2$ and $\sum_i\mu_0(\mathcal B_i)=O(d^2)$, we have 
$\sum_{i:\,\fun c_i\text{ is sparse}}\mu_1(\fun c_i)
\le O(d^2)2^{-\Omega(d)}+k e^{-\Omega(\sqrt m)}$.

Combining the  corruption bounds of regular and sparse pieces,  for any cylinder intersection $\fun c$,
\[
\mu_1(\fun c)\le \alpha \cdot  \mu_0(\fun c) +e^{-\beta m^{1/16}},
\]
where $\alpha>1,\beta>0$ are  constants.

\paragraph{Step 4. Deduce the communication lower bound.}
Fix a randomized protocol with error $\le 1/3$ and cost $t$. 
Summing over all possible $2^{t}$ cylinder intersectionsm, we have  
\[ 2/3\le  \sum
\mu_1(\fun c)\le 
\alpha \cdot \sum \mu_0(\fun c) + 2^{t}\cdot e^{-\beta m^{1/16}} \le  1/3+ 2^{t}\cdot e^{-\beta m^{1/16}} .
\]
Hence $t=\Omega(m^{1/16})=\Omega((n/\log n)^{1/32})$.

\section{Preliminaries}\label{sec:preliminary}

\paragraph{Notation.}
We use standard matrix notation. The norms $\|\cdot\|_2$,
$\|\cdot\|_{\mathrm F}$, and $\|\cdot\|_{\operatorname{op}}$ denote the
Euclidean, Frobenius, and operator norms, respectively.
Write $\ugroup(m)$ and $\sugroup(m)$ for the unitary and special unitary
groups.

For $\delta\ge0$, denote the set of matrices close to unitary by 
\[
\mathcal M_m(\delta)
=\{M\in\C^{m\times m}:\|M^\dagger M-I_m\|_{\mathrm F}\le\delta\}.
\]
Let $\nu$ denote Haar probability measure on $\sugroup(m)$, and
$\nu^{\otimes q}$ its $q$-fold product. Haar measure is invariant under
left and right multiplication and inversion.
For a probability measure $\mu$, write
\[
\mu(\fun u)=\E_\mu[\fun u],\qquad
\mu(\mathcal S)=\mu(\mathbf1_{\mathcal S}),\qquad
\|\fun u\|_p=\bigl(\E_\mu[|\fun u|^p]\bigr)^{1/p}
\quad(1\le p<\infty),
\]
whenever these expectations are defined.
For a nonempty finite set $\mathcal S$, $x\sim\mathcal S$ means uniform
sampling. Samples are independent unless a joint distribution or
conditioning is specified.

\paragraph{Communication complexity.}
We use the three-player number-on-forehead (NOF)
model~\cite{chandra1983multi,kushilevitz1997communication}.
For inputs $\bx,\by,\bz\in\{0,1\}^n$, Alice, Bob, and Charlie see
$(\by,\bz)$, $(\bx,\bz)$, and $(\bx,\by)$, respectively.
Classical protocols allow public randomness and unrestricted blackboard
interaction. Local computation is free, and communication cost is the
worst-case total number of bits or qubits exchanged. Unless stated
otherwise, protocols have error at most $1/3$ on every promised input.

Fixing the random coins of a classical NOF protocol of cost at most $t$
on all inputs partitions its input space into at most $2^t$
\emph{cylinder intersections}, one for each transcript. Their indicators
have the form
\[
\fun c(\bx,\by,\bz)=\fun f(\by,\bz)\fun g(\bx,\bz)\fun h(\bx,\by),
\]
where the factors are measurable Boolean functions of the indicated arguments.

\paragraph{Input representations.}
When bit strings are given as input, players interpret them as complex matrices. Let $b=\lceil \log n \rceil $, use $b+2$ bits to represent a number in
$[-2,2)$ with precision $2^{-b}$. For example, given a bit string $\bx\in \{0,1\}^{n}$, divide it into blocks of  length $b+2$. Each block $\ba = (\ba_1,...,\ba_{b+2} )\in \{0,1\}^{b+2}$ is decoded as a rational number in two’s-complement representation, i.e.,  \[
\mathsf{dec}(\ba) = -2\cdot\ba_{b+2} + \sum_{r=1}^{b+1} 2^{r-1-b}\cdot \ba_{r} \in [-2,2)\cap 2^{-b}\cdot\mathbb{Z}.
\]
Each consecutive pair of blocks forms a complex number, with the first as the real part and the second as the imaginary part. 
Decode the bit string from left to right into complex numbers, then arrange them row by row into  $m\times m$ matrices, where $m=\Theta(\sqrt{n/\log n})$.
Ignoring redundant bits, decode $\bx$ into
$X=(A_1,C_1,A_2,C_2)$, and similarly decode $\by$ and $\bz$
into $Y=(B_1,D_1,B_2,D_2)$ and $Z=(E_1,E_2)$.

\begin{definition}
[IUP function]
\label{def:pro}
\label{def:ABCDE}
    Given $\bx,\by,\bz\in \{0,1\}^{n}$, in the NOF view, Alice sees $(\by,\bz)$, Bob sees $(\bx,\bz)$ and  Charlie sees  $(\bx,\by)$. They decode the matrices in their respective views, obtaining
    $(Y,Z)$, $(X,Z)$, and $(X,Y)$, respectively.  Define  \[
    \fun{F}_{n}(\bx,\by,\bz)=\begin{cases}
        0, & \text{ if } \operatorname{Re}\trace(W) /m\le 0.1 \;\;\&\;\; X,Y,Z \subseteq \mathcal M_m(10^{-4});\\
         1, & \text{ if } \operatorname{Re}\trace(W) /m\ge 0.9 \;\;\&\;\; X,Y,Z \subseteq \mathcal M_m(10^{-4});\\
         *, & \text{ Otherwise.}
    \end{cases}
    \]
Under the promise that $\fun F_n\neq *$, the players communicate to determine whether $\fun F_n$ is $0$ or $1$.
\end{definition}

\paragraph{Quantum communication.}
We follow the standard quantum formalism~\cite{nielsen2010quantum}.
Our protocol uses no prior entanglement and passes a register consisting
of a control qubit and an $m$-dimensional target with basis $\{\ket j:j\in[m]\}$. We use $\ket+=\frac{\ket0+\ket1}{\sqrt2}$ and $ \ket-=\frac{\ket1-\ket0}{\sqrt2}$.
For $U\in\ugroup(m)$, define the controlled operation $\fun{ctrl}(U)=\ket0\bra0\otimes I_m+\ket1\bra1\otimes U$. In the one-clean-qubit communication model~\cite{klauck2018power},
the initial state is $\varrho_0=\ket0\bra0\otimes I_{2^q}/2^q$
for an integer $q\ge0$. Players apply input-dependent unitaries and exchange these qubits, without additional quantum memory, before a fixed projective measurement. We allow a constant number of independent runs, each using one clean qubit,
as specified in the introduction.

\section{Quantum protocol}
\label{sec:quantum_protocol}
In this section, we prove Theorem~\ref{thm:intro-quantum}
by giving an $O(\log n)$-qubit protocol for the IUP problem. To be concise, 
let $P_1,P_2,P_3$ denote Alice, Bob, and Charlie, respectively.
After decoding, their views are $(Y,Z)$, $(X,Z)$, and $(X,Y)$, where
\[
X=(A_1,C_1,A_2,C_2),\qquad
Y=(B_1,D_1,B_2,D_2),\qquad
Z=(E_1,E_2).
\]
And Recall that $$W=A_1B_1C_1D_1E_1A_2B_2C_2D_2E_2.$$ 
We first replace each input matrix $U$ by its unitary polar factor
$\widetilde U=U(U^\dagger U)^{-1/2}$, and let $\widetilde W$ be the
product of these factors in the same order as in $W$.
The promise $U\in\mathcal M_m(\sigma)$, with $\sigma=10^{-4}$,
ensures that this replacement preserves the trace gap.

\begin{claim}
\label{claim:unitarization}
For every input matrix $U$, $\widetilde U^\dagger\widetilde U=I_m$, $\|\widetilde U-U\|_{\operatorname{op}}\le\sigma$. Moreover,
\[
\frac{|\operatorname{Re}\trace\widetilde W-\operatorname{Re}\trace W|}{m}
\le 10\sigma(1+\sigma)^{9/2}.
\]
\end{claim}

Claim~\ref{claim:unitarization} is proved in Appendix~\ref{app:additional-proofs}.
Using the register and controlled operations defined in
Section~\ref{sec:preliminary}, the players execute the following protocol.

\begin{algorithm}[H]
\caption{Quantum protocol for gap trace estimation}
\label{protocol}
\begin{algorithmic}[1]
\State Each player computes $\widetilde U=U(U^\dagger U)^{-1/2}$ for every input matrix $U$ in their view.
\State $L\gets100$.
\For{$\ell=1,\ldots,L$}
  \State $P_2$ samples $j_\ell$ uniformly from $[m]$, independently of earlier trials, and prepares $\ket\Phi\gets\ket+\ket{j_\ell}$.
  \For{$r=1,\ldots,10$}
    \State $P_{i_r}$ applies $\fun{ctrl}(\widetilde U_r)$ to $\ket\Phi$, according to the table:
    \Statex \;\;\; \;\;\;\;\; \(\displaystyle
      \begin{array}{c|cccccccccc}
      r   &1&2&3&4&5&6&7&8&9&10\\ \hline
      U_r &E_2&D_2&C_2&B_2&A_2&E_1&D_1&C_1&B_1&A_1\\
      i_r &2&1&3&1&3&2&1&3&1&3
      \end{array}\)
    \If{$r<10$}
      \State $P_{i_r}$ sends the register to $P_{i_{r+1}}$.
    \EndIf
  \EndFor
  \State $P_3$ measures $\ket{\Phi}$ 
  with $\{ \ket{+}\bra{+}\otimes I_{m} ,\ket{-}\bra{-}\otimes I_{m}\}$, and let $ \omega_\ell = \begin{cases}
      1, & \ket{+}\bra{+}\otimes I_{m} ,\\
      -1, & \ket{-}\bra{-}\otimes I_{m};
  \end{cases}$ 
\EndFor
\State $P_3$ outputs $\mathbf1\!\left[L^{-1}\sum_{\ell=1}^L\omega_\ell>1/2\right]$.
\end{algorithmic}
\end{algorithm}

Each controlled operation is performed by a player who sees its matrix.
Averaging over $j_\ell$, the initial message state is  
$\ket+\bra+\otimes I_m/m$.
The operations apply the factors from right to left in $\widetilde W$,
so, conditioned on the sampled  $j_\ell \in [m]$, the final state is
\[
\ket\Phi
=\frac{\ket0\ket{j_\ell}+\ket1\widetilde W\ket{j_\ell}}{\sqrt2}.
\]
Thus,
\begin{align*}
\Pr[\omega_\ell=1\mid j_\ell]
&=\frac14\bigl\|\ket{j_\ell}+\widetilde W\ket{j_\ell}\bigr\|_2^2
 =\frac{1+\operatorname{Re}\widetilde W_{j_\ell j_\ell}}2, \text{ and }
\Pr[\omega_\ell=-1\mid j_\ell]
=\frac{1-\operatorname{Re}\widetilde W_{j_\ell j_\ell}}2.
\end{align*}
Taking the expectation over $j_\ell$ gives
\[
\E[\omega_\ell]
=\frac1m\sum_{j=1}^m\operatorname{Re}\widetilde W_{jj}
=\frac{\operatorname{Re}\trace\widetilde W}{m}.
\]
By Claim~\ref{claim:unitarization} and  taking  $\sigma = 10^{-4}$, replacing the input matrices by
their polar factors changes the normalized trace by less than $0.002$.
Thus,
\[
\fun F_n=1\ \Longrightarrow\ \E[\omega_\ell]>0.898,
\qquad
\fun F_n=0\ \Longrightarrow\ \E[\omega_\ell]<0.102.
\]
Notice that  $j_{1},...,j_{L}$ are  independent and $\Var(\omega_\ell)\le1$.
For $L=100$, Chebyshev's inequality therefore yields
\[
\Pr\!\left[
\mathbf1\!\left[\frac1L\sum_{\ell=1}^L\omega_\ell>\frac12\right]
\ne\fun F_n
\right]
\le\frac13.
\]

Finally, each iteration uses nine message transmissions, each consisting of
$1+\lceil\log_2m\rceil$ qubits. Recall that $m=\Theta(\sqrt{n/\log n})$; hence  the total communication cost is
$O(\log n)$.

\paragraph{Remark:}
The $O(\log n)$ quantum upper bound also holds in the more restricted number-in-hand (NIH) model, where $P_1$, $P_2$, and $P_3$ hold $\bx$, $\by$, and $\bz$, respectively.
Run Protocol~\ref{protocol} in the same matrix order, assigning each controlled operation to the player holding its matrix; $P_3$ initializes the register and $P_1$ measures and outputs the answer.
The communication cost and error probability are unchanged. This  quantum protocol for IUP does not rely on the overlap of inputs. Together with Theorem~\ref{thm:intro-classical}, this establishes
an exponential separation even between quantum protocols in the NIH model
and randomized protocols in the NOF model.
\section{Randomized lower bound}
\label{sec:lowerbound}
In this section, we prove Theorem~\ref{thm:intro-classical}, establishing the claimed randomized communication lower bound. The high level strategy is to define two input distributions under which the IUP function is $0$ and $1$ with high probability, respectively; then establish a corruption bound showing that distinguishing these distributions with significant  constant advantage requires $\Omega((n/\log n)^{1/32})$ communication cost.  Moreover, we adapt the regularity decomposition of Abboud, Fischer, Kelley, Lovett, and Meka~\cite{abboud2024new} to split cylinder intersections into sparse parts and regular (pseudorandom-like) parts.
Then we combine the regularity guarantees with the matrix-product estimate of Arunachalam, Girish, and Lifshitz~\cite{arunachalam2023one} to bound these parts separately and obtain the corruption bound.

\medskip
At first, we define  two distributions $\mu_1$ and $\mu_0$ of $(X,Y,Z)$ on
$\sugroup(m)^4\times\sugroup(m)^4\times\sugroup(m)^2$, where  $\sugroup(m)$ denotes the special unitary group. 
Let $\nu$ be the Haar measure of $\sugroup(m)$.
\begin{enumerate}
    \item $\mu_0$: For each $U\in \{ A_1,B_1,C_1,D_1,E_1, A_2,B_2,C_2,D_2,E_2 \}$, sample them independently from $\nu$.
    \item $\mu_{1}$: For all the  matrices except $E_2$, sample them i.i.d by $\nu$. Then let \[P= A_1 B_1  C_1 D_1, \quad Q =  A_2 B_2 C_2 D_2, \quad \text{and}\;\; E_2 = (P E_1 Q )^{-1}.\]
\end{enumerate}

Hence under $\mu_1$, $W=I_m$; under $\mu_0$, $\operatorname{Re}\trace(W)/m$ is close to zero with high probability. Haar invariance gives the following marginal identities.

\begin{claim}\label{claim:haar_measure_marginal}
$\mu_0$ and $\mu_1$  have identical marginals on each NOF view, i.e.,\[
(X,Y) \sim \nu^{\otimes 4}\times \nu^{\otimes 4},
\qquad
(X,Z) \sim \nu^{\otimes 4}\times \nu^{\otimes 2},
\qquad
(Y,Z) \sim \nu^{\otimes 4}\times \nu^{\otimes 2}.
\]
\end{claim}
The proof of Claim~\ref{claim:haar_measure_marginal} is in Appendix~\ref{app:additional-proofs}.

\medskip
To apply a protocol on bit-string inputs to Haar-distributed matrices,  
we introduce the following round process, which maps the coordinates of a continuous complex matrix to the decoding grid $[-2,2)\cap 2^{-b}\cdot\mathbb{Z}$. 
For any measurable function $\fun u$ and set $\mathcal S$, write $\mu_j(\fun u):=\E_{\mu_j}[\fun u]$ and $\mu_j(\mathcal S):=\mu_j(\mathbf1_{\mathcal S})$ for $j\in\{0,1\}$.

\paragraph{Rounding Process.}
Recall that $b=\lceil\log n\rceil$.
Replace every real or imaginary coordinate $t\in[-1,1]$ by
\[
2^{-b}\left\lceil 2^bt-\frac12\right\rceil\in[-2,2)\cap2^{-b}\mathbb Z.
\]
This is nearest-grid rounding, with ties rounded down. Encode these coordinates using the inverse of $\mathsf{dec}$, in the decoding order specified in Section~\ref{sec:preliminary}, and set unused bits to zero.
Denote the resulting encoding maps by
$\fun R_x:\sugroup(m)^4\to\{0,1\}^n$,
$\fun R_y:\sugroup(m)^4\to\{0,1\}^n$, and
$\fun R_z:\sugroup(m)^2\to\{0,1\}^n$; write
$\fun R(X,Y,Z)=(\fun R_x(X),\fun R_y(Y),\fun R_z(Z))$.
For $j\in\{0,1\}$, let $\widetilde\mu_j$ be the distribution of $\fun R(X,Y,Z)$ when $(X,Y,Z)\sim\mu_j$.
Then, for every Boolean input set $\mathcal S$,
\begin{equation}\label{eq:rounding-pullback}
\widetilde\mu_j(\mathcal S)=\mu_j(\fun R^{-1}(\mathcal S))\qquad(j=0,1).
\end{equation}

\begin{claim}\label{claim:rounding}
For all sufficiently large $n$, the inputs under $\mu_0$ and $\mu_1$ satisfy
\[
\Pr_{\mu_1}[\operatorname{Re}\trace(W)/m\ge0.9]=1,\qquad
\Pr_{\mu_0}[\operatorname{Re}\trace(W)/m\le0.1]\ge1-\frac1m.
\]
After the rounding process, every matrix belongs to $\mathcal M_m(10^{-4})$, and
\[
\Pr_{\widetilde\mu_1}[\fun F_n=1]=1,\qquad
\Pr_{\widetilde\mu_0}[\fun F_n=0]\ge1-\frac1m.
\]
\end{claim}

The proof of Claim~\ref{claim:rounding} is in Appendix~\ref{app:additional-proofs}.
We can therefore work with Haar matrices and apply the Boolean protocol after rounding.
Equation~\eqref{eq:rounding-pullback} preserves every acceptance probability exactly.

\medskip

\paragraph{Cylinder intersection.}
Let $\mathcal{X}$, $\mathcal{Y}$, and $\mathcal{Z}$ be measurable spaces.
Given measurable Boolean functions
\[
\fun f : \mathcal{Y} \times \mathcal{Z} \to \{0,1\},\qquad
\fun g : \mathcal{X} \times \mathcal{Z} \to \{0,1\},\qquad
\fun h : \mathcal{X} \times \mathcal{Y} \to \{0,1\},
\]
define the cylinder intersection indicator by
\[
\fun c : \mathcal{X} \times \mathcal{Y} \times \mathcal{Z} \to \{0,1\},
\qquad
\fun c(X,Y,Z) := \fun f(Y,Z)\cdot \fun g(X,Z)\cdot \fun h(X,Y).
\]

The randomized   lower bound follows directly from the following corruption bound, which shows any short classic protocol can not distinguish ${\mu}_0$ and ${\mu}_1$ significantly on  cylinder intersections.

\begin{lemma}[corruption bound]
\label{lem:corruption_bound}
 Let $\mathcal{X}=\mathcal{Y} = \sugroup(m)^{4}$ and $\mathcal{Z} = \sugroup(m)^{2}$. 
  There exist universal constants $\alpha>1$ and $\beta>0$ such that for each cylinder intersection  $\fun c$, we have \[
    \mu_{1}(\fun c)\le \alpha \cdot\mu_{0}(\fun c) + \mathrm{exp}\big(-\beta m^{1/16} \big).
    \]
\end{lemma}

The proof of  Lemma~\ref{lem:corruption_bound} is deferred to  Section~\ref{sec:corruption_bound}. We are now ready to prove Theorem~\ref{thm:intro-classical}.

\begin{proof}[Proof of Theorem~\ref{thm:intro-classical}]
By constant amplification and truncation, take a protocol $\Pi$ with cost at most $t$ on all inputs, and error at most  $\eta=1/[4(\alpha+1)]$, where $\alpha$ is the constant in the corruption bound. The amplification changes the cost by only a constant factor. Let $\rho$ denote the random coins. Then
\[
\Pr_\rho[\Pi_\rho(\bx,\by,\bz)\ne\fun F_n(\bx,\by,\bz)]\le\eta
\qquad\text{whenever }\fun F_n(\bx,\by,\bz)\ne *.
\]
Fix $\rho$; let $s\le2^t$ be the number of accepting transcript leaves, and let $\widetilde{\fun c}_\ell$ be the indicator of the $\ell$-th leaf, for $\ell\in[s]$. Distinct leaves correspond to disjoint sets of inputs. Each $\widetilde{\fun c}_\ell$ is a Boolean cylinder intersection. Since the rounding process  $\fun R$ encodes $X,Y,Z$ separately, $\fun c_\ell:=\widetilde{\fun c}_\ell\circ\fun R$ is also a cylinder intersection.

The probability identity~\eqref{eq:rounding-pullback} gives $\widetilde\mu_j(\widetilde{\fun c}_\ell)=\mu_j(\fun c_\ell)$ for $j=0,1$. Apply the corruption bound of  Lemma~\ref{lem:corruption_bound}:
\[
\widetilde\mu_1(\widetilde{\fun c}_\ell)
=\mu_1(\fun c_\ell)
\le\alpha\mu_0(\fun c_\ell)+e^{-\beta m^{1/16}}
=\alpha\widetilde\mu_0(\widetilde{\fun c}_\ell)+e^{-\beta m^{1/16}}.
\]
Summing over the accepting leaves and averaging over $\rho$ yields
\[
\Pr_{\widetilde\mu_1,\rho}[\Pi_\rho=1]
\le\alpha\Pr_{\widetilde\mu_0,\rho}[\Pi_\rho=1]+2^t e^{-\beta m^{1/16}}.
\]
Claim~\ref{claim:rounding} says that $\widetilde\mu_1$ is supported on $1$-inputs and $\widetilde\mu_0$ assigns probability at most $1/m$ to inputs with $\fun F_n\ne0$. Together with the error bound above, this gives
\[
\Pr_{\widetilde\mu_1,\rho}[\Pi_\rho=1]\ge1-\eta,\qquad
\Pr_{\widetilde\mu_0,\rho}[\Pi_\rho=1]\le\eta\Pr_{\widetilde\mu_0}[\fun F_n=0]+\Pr_{\widetilde\mu_0}[\fun F_n\ne0]\le\eta+\frac1m.
\]
Hence, for sufficiently large $m$,
\[
2^t e^{-\beta m^{1/16}}\ge1-(\alpha+1)\eta-\frac\alpha m\ge\frac12,
\qquad t\ge\frac\beta{\ln2}m^{1/16}-1.
\]
Therefore the  protocol  must cost $t= \Omega(m^{1/16})=\Omega((n/\log n)^{1/32})$ bits.
\end{proof}

\subsection{Corruption bound via regularity decomposition}
\label{sec:corruption_bound}
In this section, we prove  Lemma~\ref{lem:corruption_bound}.
We first prove the bound for finite-type cylinder intersections,
whose factors are constant on finite equal-measure partitions,
so that we can apply the matrix decomposition of
Abboud, Fischer, Kelley, Lovett, and Meka~\cite{abboud2024new}.
We then extend the bound to all measurable cylinder intersections
by approximation.

\medskip
At first, we define the \textit{finite-type}  cylinder intersection. 

\begin{definition}[finite type]\label{def:finite_type}
A cylinder intersection $\fun c=\fun f\fun g\fun h$ is \textit{finite-type} if there are finite measurable partitions $\mathcal P_X,\mathcal P_Y,\mathcal P_Z$ of $\mathcal X,\mathcal Y,\mathcal Z$, respectively, such that for every $\mathcal X'\in\mathcal P_X$, $\mathcal Y'\in\mathcal P_Y$, $\mathcal Z'\in\mathcal P_Z$,
\[
\nu^{\otimes4}(\mathcal X')=\frac1{|\mathcal P_X|},\qquad
\nu^{\otimes4}(\mathcal Y')=\frac1{|\mathcal P_Y|},\qquad
\nu^{\otimes2}(\mathcal Z')=\frac1{|\mathcal P_Z|},
\]
and the restrictions
\[
\fun f|_{\mathcal Y'\times\mathcal Z'},\qquad
\fun g|_{\mathcal X'\times\mathcal Z'},\qquad
\fun h|_{\mathcal X'\times\mathcal Y'}
\quad\text{are constant.}
\]
All    averages of  functions on the cells $\mathcal X'$ ,$\mathcal Y'$, $\mathcal Z'$ or a union of cells use \textit{normalized} Haar measure.
\end{definition}

\medskip
The following claim shows that every measurable cylinder intersection admits a finite-type approximation.
\begin{claim}
\label{claim:step_function_arppox}
For any measurable cylinder intersection indicator $\fun c: \mathcal{X}\times \mathcal{Y}\times \mathcal{Z}\to \{0,1\}$ and any integer $r\ge 1$, there exists a finite-type cylinder intersection $\fun c_r$ such that
$|\mu_i(\fun c-\fun c_r)|\le1/r$ for $i=0,1$.
\end{claim}

The proof of Claim~\ref{claim:step_function_arppox} is in Appendix~\ref{app:additional-proofs}.

\medskip
To prove Lemma~\ref{lem:corruption_bound}, we follow the structure-versus-pseudorandomness approach: we decompose cylinder intersections into sparse parts and pseudorandom parts, and bound them separately. We use grid regularity, introduced in the context of communication complexity by Kelley, Lovett, and Meka~\cite{kelley2024explicit}, building on the work of Kelley and Meka~\cite{kelley2023strong} on 3-term arithmetic progressions. We obtain the regularity  decomposition on cylinder intersections by adapting the results  of Abboud, Fischer, Kelley, Lovett, and Meka~\cite{abboud2024new}.

\begin{definition}
    [Grid norm, grid regularity, and min-degree]
    Let $\mathcal S,\mathcal T$ be probability spaces. For a non-negative function $\fun u:\mathcal S\times\mathcal T\to\R_+$ and integers $k,\ell\ge1$, define the \textit{grid norm}\footnote{We follow the notation of Kelley, Lovett, and Meka~\cite{kelley2024explicit} and Abboud, Fischer, Kelley, Lovett, and Meka~\cite{abboud2024new}. The letter $U$ stands for uniformity, as in the Gowers uniformity norms of additive combinatorics.} by
    \[
    \|\fun u\|_{U(k,\ell)}
    :=\left(\E_{\substack{s_1,\ldots,s_k\sim\mathcal S\\t_1,\ldots,t_\ell\sim\mathcal T}}
    \left[\prod_{a=1}^k\prod_{b=1}^{\ell}\fun u(s_a,t_b)\right]\right)^{1/(k\ell)},
    \]
    where all samples are independent. We use $k=2$ throughout; for an integer $d\ge1$,
    \[
    \|\fun u\|_{U(2,d)} = \left(\E_{s,s'\sim\mathcal S}\!\left[\left(\E_{t\sim\mathcal T}[\fun u(s,t)\fun u(s',t)]\right)^d\right]\right)^{1/(2d)}.
    \]
Let $\delta := \E_{s\sim\mathcal{S},t\sim \mathcal{T}}[\fun{u}(s,t)]$. Given $\epsilon>0$,  $\fun u$ is \textit{$(\epsilon,d)$-regular} if
\footnote{We write $(\epsilon,d)$-regular for $(\epsilon,2,d)$-regular in Abboud, Fischer, Kelley, Lovett, and Meka~\cite{abboud2024new}, since only the case $k=2$ is used here.}
\[
    \|\fun u\|_{U(2,d)} \le (1+\epsilon)\cdot\delta. 
\] 
The function $\fun u$ satisfies \textit{$\epsilon$-min-degree}  if  \[
\text{for any } s\in \mathcal{S},\quad 
\E_{t\in \mathcal{T}} [\fun{u}(s,t)] \ge (1-\epsilon)\cdot \delta.
\]
\end{definition}

In the following argument, we refer to grid regularity simply as regularity.

\begin{theorem}[Regularity decomposition, adapted from \cite{abboud2024new}]
\label{thm:regularity_decomposition}
Fix $\epsilon>0$ and an integer $d\ge1$.  
For any finite-type  cylinder intersection $\fun c = \fun{f} \fun{g} \fun{h}$ on $ \mathcal{X}\times\mathcal{Y}\times \mathcal{Z}$, there exist $k$ sets
$\mathcal B_i=\mathcal X_i\times\mathcal Y_i\times\mathcal Z_i
\subseteq\mathcal X\times\mathcal Y\times\mathcal Z$,
indexed by $i\in[k]$, where $\mathcal X_i,\mathcal Y_i,\mathcal Z_i$ are nonempty unions of cells
from their respective finite-type  partitions,  satisfying 
\begin{enumerate}
 \item For each $i\in [k]$, there exist Boolean functions $\fun f_i$ and $\fun g_i$, which are constant on products of the partition cells and supported  on $\mathcal{Y}_{i}\times \mathcal{Z}_{i}$ and $\mathcal{X}_{i}\times \mathcal{Z}_{i}$, s.t. $\fun c = \sum_{i=1}^{k} \fun f_{i} \fun g_{i} \fun h$; \footnote{Both functions are defined on the full input spaces and
vanish outside the indicated subsets. They need not be restrictions
of $\fun f$ and $\fun g$.}
    \item $\sum_{i=1}^{k}\mu_{0}(\mathcal{B}_{i}) \le O(d^{2})$ and $ k\le \mathrm{exp}({c_{\epsilon}d^{7}})$, where  $c_{\epsilon}$ is a  constant only depending on $\epsilon$; 
    \item For each $ i\in [k]$, let $\delta_{\fun f,i} := \E_{Y\sim \mathcal{Y}_{i}, Z\sim \mathcal{Z}_{i}}[\fun f_{i}(Y,Z)] $
 and $\delta_{\fun g,i} := \E_{X\sim \mathcal{X}_{i}, Z\sim \mathcal{Z}_{i}}[\fun g_{i}(X,Z)]$. \footnote{$\delta_{\fun f,i}$ and $\delta_{\fun g,i}$ use normalized Haar measure  on $\mathcal{Y}_i\times \mathcal{Z}_i$ and $\mathcal{X}_i\times \mathcal{Z}_i$; this expectation does not necessarily equal the expectation on the whole input space $\mathcal{Y}\times \mathcal{Z}$ or $\mathcal{X}\times \mathcal{Z}$.} Either \begin{enumerate}
     \item $\delta_{\fun f,i}\le 2^{-d}$ or $\delta_{\fun g,i}\le 2^{-d}$;
     \item or both $\fun f_i$ and $\fun g_i$ are $(\epsilon,d)$-regular
and satisfy $\epsilon$-min-degree on $\mathcal Y_i\times\mathcal Z_i$
and $\mathcal X_i\times\mathcal Z_i$, respectively.\footnote{Regularity and min-degree use normalized Haar measure
on the indicated subsets $\mathcal Y_i\times\mathcal Z_i$
and $\mathcal X_i\times\mathcal Z_i$.}
 \end{enumerate}
\end{enumerate}
\end{theorem}

The proof of Theorem~\ref{thm:regularity_decomposition} is in Appendix~\ref{app:proofs}. Since $\fun f_{i}$ and $\fun g_{i}$ are $0$ outside the support, let  $\fun h_i := \fun h \cdot \mathbf{1}_{\mathcal{X}_{i}}\cdot \mathbf{1}_{\mathcal{Y}_{i}}$ and $\fun c_i := \fun f_i \fun g_i \fun h_i$; then  $\fun c = \sum_{i=1}^{k} \fun c_{i}$. We further define $\delta_{\fun h,i} := \E_{x\sim \mathcal{X}_{i}, y\sim \mathcal{Y}_{i}}[\fun h_{i}(x,y)]$.

\medskip

With Theorem~\ref{thm:regularity_decomposition} in hand, we treat $\{\fun c_i\}_{i=1}^{k}$ in two cases: \textit{regular}  and  \textit{sparse}. Define the regular set and sparse set  as \[
\mathcal{I}_{r}:= \{ i\in [k]: \delta_{\fun f,i}> 2^{-d}, \delta_{\fun g,i}> 2^{-d}, \delta_{\fun h,i}> 2^{-\epsilon d} \},\quad \text{and}\;\; \mathcal{I}_{s}:=[k]\setminus \mathcal{I}_{r}.
\]
For the next two lemmas, fix a sufficiently small constant $\epsilon>0$ and a sufficiently large integer $d$.

\begin{lemma}[Regular case]
\label{lem:regulary_corruption_bound}
    There exist universal constants $\kappa>0$ and $\alpha>1$ such that for any $i\in \mathcal{I}_{r}$, we have \[
    \mu_{1}(\fun c_{i})\le \alpha
    \cdot \mu_{0}(\fun c_{i})+  \mathrm{exp}(-\kappa\sqrt{m}/d).
    \]
\end{lemma}

\begin{lemma}[Sparse case]\label{lem:sparse_currption_bound}
There exist  universal constants $C,\kappa>0$ such that
\[
\sum_{i\in\mathcal I_s}\mu_1(\fun c_i)
\le Cd^2 2^{-\epsilon d/2}+|\mathcal I_s|e^{-\kappa\sqrt m/2}.
\]
\end{lemma}

We prove Lemma~\ref{lem:regulary_corruption_bound} and Lemma~\ref{lem:sparse_currption_bound} in Section~\ref{sec:regular_case} and Section~\ref{sec:sparse_case} respectively. Now combining these two lemmas we are ready to  prove   Lemma~\ref{lem:corruption_bound}.

\begin{proof}[Proof of Lemma~\ref{lem:corruption_bound}]
Fix $\epsilon=10^{-4}$ and first suppose $\fun c$ is  finite-type. Theorem~\ref{thm:regularity_decomposition} gives Boolean functions $\fun c_i$ with $\fun c=\sum_{i=1}^k\fun c_i$ and $k\le e^{c_\epsilon d^7}$. Their supports are disjoint because $\fun c$ is Boolean. Lemma~\ref{lem:regulary_corruption_bound} compares each regular term with its $\mu_0$-mass, while Lemma~\ref{lem:sparse_currption_bound} bounds the total sparse mass. Taking a common smaller $\kappa>0$ and $C\ge1$, we obtain
\begin{align*}
\mu_1(\fun c)
&=\sum_{i\in\mathcal I_r}\mu_1(\fun c_i)+\sum_{i\in\mathcal I_s}\mu_1(\fun c_i)\\
&\le\alpha\sum_{i\in\mathcal I_r}\mu_0(\fun c_i)+Cd^2 2^{-\epsilon d/2}
+|\mathcal I_r|e^{-\kappa\sqrt m/d}+|\mathcal I_s|e^{-\kappa\sqrt m/2}\\
&\le\alpha\mu_0(\fun c)+Cd^2 2^{-\epsilon d/2}
+\exp(c_\epsilon d^7-\kappa\sqrt m/d).
\end{align*}
Choose a fixed $\gamma>0$ with $c_\epsilon\gamma^8\le\kappa/2$, and put $d=\lfloor\gamma m^{1/16}\rfloor$. For sufficiently large $m$, $d$ satisfies the hypotheses of both lemmas, and
\[
c_\epsilon d^7\le\frac{\kappa\sqrt m}{2d},\qquad
Cd^2 2^{-\epsilon d/2}+\exp\!\left(-\frac{\kappa\sqrt m}{2d}\right)
\le\exp(-\beta m^{1/16})
\]
for a universal $\beta>0$.

For any measurable $\fun c$, Claim~\ref{claim:step_function_arppox} supplies finite-type $\fun c_r$ with $|\mu_j(\fun c-\fun c_r)|\le1/r$ for $j=0,1$. Applying the bound just proved to $\fun c_r$ gives
\[
\mu_1(\fun c)\le\mu_1(\fun c_r)+\frac1r
\le\alpha\mu_0(\fun c)+\frac{\alpha+1}{r}+e^{-\beta m^{1/16}}.
\]
Putting $r\to\infty$ completes the proof. 
\end{proof}

\subsubsection{Corruption bound on regular cylinder intersections.}
\label{sec:regular_case}

In this section, we prove Lemma~\ref{lem:regulary_corruption_bound}, i.e., the corruption bound of the regular part. Let $\mathcal{X}=\mathcal{Y} = \sugroup(m)^{4}$ and $\mathcal{Z} = \sugroup(m)^{2}$. Recall that $\fun f:\mathcal{Y}\times \mathcal{Z}\to \{0,1\}$, $\fun g:\mathcal{X}\times \mathcal{Z}\to \{0,1\}$ and $\fun h:\mathcal{X}\times \mathcal{Y}\to \{0,1\}$; $\nu$ is the Haar measure of $\sugroup(m)$;  by Claim~\ref{claim:haar_measure_marginal}, for both  $\mu_0$ and  $\mu_1$ we have  \[
(X,Y) \sim \nu^{\otimes 4}\times \nu^{\otimes 4},
\qquad
(X,Z) \sim \nu^{\otimes 4}\times \nu^{\otimes 2},
\qquad
(Y,Z) \sim \nu^{\otimes 4}\times \nu^{\otimes 2}.
\]

We first bound $\mu_1(\fun c_i)$ by the grid norms of $\fun f_i,\fun g_i$; this bound will also handle the sparse terms. For regular terms, a moment bound then gives a matching lower bound on $\mu_0(\fun c_i)$.

\medskip
At first, 
we derive the following grid-norm bound from
Arunachalam, Girish, and Lifshitz~\cite[Lemma~5.2]{arunachalam2023one} using a moment argument and
Cauchy--Schwarz, which might be of independent interest.
\begin{theorem}[Grid norm bound] 
\label{thm:Grid_norm}
There exist universal constants $\lambda>1$ and $\kappa>0$ such that,
for any $\fun f:\mathcal{Y}\times \mathcal{Z}\to \{0,1\}$, $\fun g:\mathcal{X}\times \mathcal{Z}\to \{0,1\}$  and integer $d\ge 2$,  we have \[
\big \|\E_{\mu_1} [\fun f(Y,Z)\cdot \fun g(X,Z)\mid X,Y] \big \|_{d}\le \lambda^{1+2/d} \|\fun f\|_{U(2,d)}\|\fun g\|_{U(2,d)} + \mathrm{exp}(-\kappa\sqrt{m}/d).
\]
    
\end{theorem}

The proof of Theorem~\ref{thm:Grid_norm} is in Appendix~\ref{app:proof-grid-norm}. We next apply this bound to each $\fun c_i=\fun f_i\fun g_i\fun h_i$. For an integer $p\ge2$, define the local grid norms using normalized Haar measure:
\[
\|\fun f_i\|_{U(2,p),i}
:=\left(\E_{Y,Y'\sim\mathcal Y_i}\!\left[\left(\E_{Z\sim\mathcal Z_i}[\fun f_i(Y,Z)\fun f_i(Y',Z)]\right)^p\right]\right)^{1/(2p)},
\]
\[
\|\fun g_i\|_{U(2,p),i}
:=\left(\E_{X,X'\sim\mathcal X_i}\!\left[\left(\E_{Z\sim\mathcal Z_i}[\fun g_i(X,Z)\fun g_i(X',Z)]\right)^p\right]\right)^{1/(2p)}.
\]
The following corollary follows from Theorem~\ref{thm:Grid_norm} by H\"older's inequality and normalization on $\mathcal B_i$.
\begin{corollary}[Grid norm bound on $\mathcal B_i$]\label{lem:local-grid-bound}
With the constants  $\lambda>1 $ and $\kappa>0$ from Theorem~\ref{thm:Grid_norm}, for every $i\in[k]$ and integer $p\ge2$,
\[
\mu_1(\fun c_i)
\le\lambda^{1+2/p}\mu_0(\mathcal B_i)\delta_{\fun h,i}^{1-1/p}
\|\fun f_i\|_{U(2,p),i}\|\fun g_i\|_{U(2,p),i}
+e^{-\kappa\sqrt m/p}.
\]
\end{corollary}
\begin{proof}[Proof of Corollary~\ref{lem:local-grid-bound}]
By Claim~\ref{claim:haar_measure_marginal}, $(X,Y)$ has product Haar measure under $\mu_1$. Conditioning on $(X,Y)$ and applying H\"older's inequality gives
\[
\mu_1(\fun c_i)
=\E_{X,Y}[\fun h_i\E_{\mu_1}[\fun f_i\fun g_i\mid X,Y]]
\le\|\fun h_i\|_{p/(p-1)}
\|\E_{\mu_1}[\fun f_i\fun g_i\mid X,Y]\|_p.
\]
Theorem~\ref{thm:Grid_norm} applies to the Boolean functions $\fun f_i,\fun g_i$ with moment $p$; using $\|\fun h_i\|_{p/(p-1)}\le1$ gives
\[
\mu_1(\fun c_i)
\le\lambda^{1+2/p}\|\fun h_i\|_{p/(p-1)}
\|\fun f_i\|_{U(2,p)}\|\fun g_i\|_{U(2,p)}
+e^{-\kappa\sqrt m/p}.
\]
All norms here use Haar measure on the full input spaces. The support restrictions give
\begin{align*}
\|\fun f_i\|_{U(2,p)}
&=\nu^{\otimes4}(\mathcal Y_i)^{1/p}\nu^{\otimes2}(\mathcal Z_i)^{1/2}\|\fun f_i\|_{U(2,p),i},\\
\|\fun g_i\|_{U(2,p)}
&=\nu^{\otimes4}(\mathcal X_i)^{1/p}\nu^{\otimes2}(\mathcal Z_i)^{1/2}\|\fun g_i\|_{U(2,p),i},\\
\|\fun h_i\|_{p/(p-1)}
&=[\nu^{\otimes4}(\mathcal X_i)\nu^{\otimes4}(\mathcal Y_i)\delta_{\fun h,i}]^{1-1/p}.
\end{align*}
Multiplying these identities yields
\[
\|\fun h_i\|_{p/(p-1)}\|\fun f_i\|_{U(2,p)}\|\fun g_i\|_{U(2,p)}
=\mu_0(\mathcal B_i)\delta_{\fun h,i}^{1-1/p}
\|\fun f_i\|_{U(2,p),i}\|\fun g_i\|_{U(2,p),i},
\]
which proves the bound.
\end{proof}

\medskip
The next theorem bounds the relative deviation of
$\E_{Z\sim\mathcal Z_i}[\fun f_i(Y,Z)\fun g_i(X,Z)]$
from $\delta_{\fun f,i}\delta_{\fun g,i}$,
using normalized Haar measure on $\mathcal X_i\times\mathcal Y_i$. The functions supplied by Theorem~\ref{thm:regularity_decomposition} are constant on finite equal-measure cells, so their normalized Haar averages equal the arithmetic means of their constant values on the equal-measure cells.
\begin{theorem}[{\cite[Lemma~4.8]{kelley2024explicit}}]\label{thm:KLM_moment}
Fix $i\in[k]$, $0<\epsilon<1/20$, and an even integer $\ell\ge2$ with $\lceil\ell/\epsilon\rceil\le d$. If the nonzero functions $\fun f_i,\fun g_i$ are $(\epsilon,d)$-regular and satisfy $\epsilon$-min-degree on $\mathcal Y_i\times\mathcal Z_i$ and $\mathcal X_i\times\mathcal Z_i$, respectively, then
\[
\left\|\frac{\E_{Z\sim\mathcal Z_i}[\fun f_i(Y,Z)\fun g_i(X,Z)]}{\delta_{\fun f,i}\delta_{\fun g,i}}-1\right\|_\ell\le20\epsilon.
\]
The norm uses normalized Haar measure on $\mathcal X_i\times\mathcal Y_i$.
\end{theorem}
\begin{remark}
    In our notation, Lemma~4.8 of Kelley, Lovett, and Meka~\cite{kelley2024explicit} requires $(\epsilon,p)$-regularity, where $p=\lceil\ell/\epsilon\rceil$. Since $p\le d$, H\"older's inequality and the assumed $(\epsilon,d)$-regularity give
\[
\|\fun f_i\|_{U(2,p),i}\le\|\fun f_i\|_{U(2,d),i}\le(1+\epsilon)\delta_{\fun f,i},
\qquad
\|\fun g_i\|_{U(2,p),i}\le\|\fun g_i\|_{U(2,d),i}\le(1+\epsilon)\delta_{\fun g,i},
\]
with normalized Haar measure on $\mathcal Y_i\times\mathcal Z_i$ and $\mathcal X_i\times\mathcal Z_i$, respectively. Thus both functions satisfy the required $(\epsilon,p)$-regularity.
\end{remark}

\medskip
Now we are ready to  prove Lemma~\ref{lem:regulary_corruption_bound}.
\begin{proof}[Proof of Lemma~\ref{lem:regulary_corruption_bound}]
Fix $0<\epsilon\le10^{-4}$, $d\ge8/\epsilon$, and $i\in\mathcal I_r$. We build a lower bound on $\mu_0(\fun c_i)$ and an upper bound on $\mu_1(\fun c_i)$, then compare them.

\paragraph{Step 1. Lower bound $\mu_0(\fun c_i)$.}
Since $\delta_{\fun f,i},\delta_{\fun g,i}>2^{-d}$, Theorem~\ref{thm:regularity_decomposition} gives regularity and min-degree on the respective subsets:
\[
\|\fun f_i\|_{U(2,d),i}\le(1+\epsilon)\delta_{\fun f,i},\qquad
\|\fun g_i\|_{U(2,d),i}\le(1+\epsilon)\delta_{\fun g,i},
\]
\[
\E_{Z\sim\mathcal Z_i}[\fun f_i(Y,Z)]\ge(1-\epsilon)\delta_{\fun f,i}\quad(\forall Y\in\mathcal Y_i),\qquad
\E_{Z\sim\mathcal Z_i}[\fun g_i(X,Z)]\ge(1-\epsilon)\delta_{\fun g,i}\quad(\forall X\in\mathcal X_i).
\]
Take $\ell=2\lfloor\epsilon d/2\rfloor$. Then $\ell\ge\epsilon d/2\ge2$ and $\lceil\ell/\epsilon\rceil\le d$, so Theorem~\ref{thm:KLM_moment} gives
\[
\left\|\frac{\E_{Z\sim\mathcal Z_i}[\fun f_i(Y,Z)\fun g_i(X,Z)]}{\delta_{\fun f,i}\delta_{\fun g,i}}-1\right\|_\ell\le20\epsilon.
\]
H\"older's inequality and $\|\fun h_i\|_{\ell/(\ell-1)}=\delta_{\fun h,i}^{1-1/\ell}$, with normalized Haar measure on $\mathcal X_i\times\mathcal Y_i$, therefore imply
\[
\left|\E_{X\sim\mathcal X_i,Y\sim\mathcal Y_i}
\left[\fun h_i(X,Y)\left(\E_{Z\sim\mathcal Z_i}[\fun f_i(Y,Z)\fun g_i(X,Z)]-\delta_{\fun f,i}\delta_{\fun g,i}\right)\right]\right|
\le20\epsilon\delta_{\fun f,i}\delta_{\fun g,i}\delta_{\fun h,i}^{1-1/\ell}.
\]
Since $\delta_{\fun h,i}>2^{-\epsilon d}$, we have $\delta_{\fun h,i}^{-1/\ell}\le4$. Hence
\begin{align}
\mu_0(\fun c_i)
&=\mu_0(\mathcal B_i)\E_{X\sim\mathcal X_i,Y\sim\mathcal Y_i}
\big[\fun h_i(X,Y)\E_{Z\sim\mathcal Z_i}[\fun f_i(Y,Z)\fun g_i(X,Z)]\big]\notag\\
&\ge\mu_0(\mathcal B_i)\delta_{\fun f,i}\delta_{\fun g,i}\delta_{\fun h,i}
(1-20\epsilon\delta_{\fun h,i}^{-1/\ell})\notag\\
&\ge(1-80\epsilon)\mu_0(\mathcal B_i)\delta_{\fun f,i}\delta_{\fun g,i}\delta_{\fun h,i}.
\label{eq:regular-independent}
\end{align}

\paragraph{Step 2. Upper bound $\mu_1(\fun c_i)$.}
Corollary~\ref{lem:local-grid-bound} bounds $\mu_1(\fun c_i)$ by the local grid norms. Taking $p=d$ and substituting the regularity bounds in Step~1 gives
\begin{equation}\label{eq:regular-constrained}
\mu_1(\fun c_i)\le\lambda^{1+2/d}(1+\epsilon)^2\mu_0(\mathcal B_i)\delta_{\fun f,i}\delta_{\fun g,i}\delta_{\fun h,i}^{1-1/d}+e^{-\kappa\sqrt m/d}.
\end{equation}

\paragraph{Step 3. Compare the bounds.}
Since $\delta_{\fun h,i}^{-1/d}<2^\epsilon$, the lower bound~\eqref{eq:regular-independent} and upper bound~\eqref{eq:regular-constrained} imply
\[
\mu_1(\fun c_i)\le\frac{\lambda^{1+2/d}(1+\epsilon)^2 2^\epsilon}{1-80\epsilon}\mu_0(\fun c_i)+e^{-\kappa\sqrt m/d}
\le4\lambda^2\mu_0(\fun c_i)+e^{-\kappa\sqrt m/d}.
\]
Thus $\alpha=4\lambda^2$ suffices.
\end{proof}

\subsubsection{Corruption bound on sparse cylinder intersections.}
\label{sec:sparse_case}

In this section, we prove Lemma~\ref{lem:sparse_currption_bound}. We use Corollary~\ref{lem:local-grid-bound} with $p=2$: for Boolean functions, the grid norm $U(2,2)$ is at most the square root of the density.

\begin{proof}[Proof of Lemma~\ref{lem:sparse_currption_bound}]
For any Boolean function $\fun u:\mathcal S\times\mathcal T\to\{0,1\}$ on probability spaces $\mathcal S,\mathcal T$, take independent samples $s,s'\sim\mathcal S$ and $t,t'\sim\mathcal T$. Since $0\le\fun u\le1$,
\[
\|\fun u\|_{U(2,2)}^4
=\E[\fun u(s,t)\fun u(s',t)\fun u(s,t')\fun u(s',t')]
\le\E[\fun u(s,t)\fun u(s',t')]
=(\E\fun u)^2.
\]
Applying this to $\fun f_i,\fun g_i$ on their respective subsets gives
\[
\|\fun f_i\|_{U(2,2),i}\le\delta_{\fun f,i}^{1/2},\qquad
\|\fun g_i\|_{U(2,2),i}\le\delta_{\fun g,i}^{1/2}.
\]
For $i\in\mathcal I_s$, at least one of $\delta_{\fun f,i},\delta_{\fun g,i},\delta_{\fun h,i}$ is at most $2^{-\epsilon d}$, since $0<\epsilon\le1$. Corollary~\ref{lem:local-grid-bound} with $p=2$ therefore yields
\[
\mu_1(\fun c_i)
\le\lambda^2\mu_0(\mathcal B_i)\sqrt{\delta_{\fun f,i}\delta_{\fun g,i}\delta_{\fun h,i}}+e^{-\kappa\sqrt m/2}
\le\lambda^2\mu_0(\mathcal B_i)2^{-\epsilon d/2}+e^{-\kappa\sqrt m/2}.
\]
Summing and using $\sum_i\mu_0(\mathcal B_i)=O(d^2)$ from Theorem~\ref{thm:regularity_decomposition} gives a universal constant $C>0$ such that
\[
\sum_{i\in\mathcal I_s}\mu_1(\fun c_i)
\le\lambda^2 2^{-\epsilon d/2}\sum_{i\in\mathcal I_s}\mu_0(\mathcal B_i)
+|\mathcal I_s|e^{-\kappa\sqrt m/2}
\le Cd^2 2^{-\epsilon d/2}+|\mathcal I_s|e^{-\kappa\sqrt m/2}.\qedhere
\]
\end{proof}

\section*{AI Declaration}
AI are used extensively for this work. Below, we list our AI usages following the AI Methodology template recommended by the
\href{https://simons.berkeley.edu/news-publications-videos/ai-tcs-working-group}
{Simons Institute's AI + TCS Working Group}.
\begin{flushleft}
\renewcommand{\arraystretch}{1.3}
\begin{tabular}{@{}l@{\qquad}c@{}}
\textbf{Activity} & \textbf{AI used} \\
Asking the research question & No \\
Coming up with the approach & Yes \\
Proof development & Yes \\
Writing and exposition & Yes \\
Checking for bugs & Yes \\
Other supporting tasks & Yes \\
\end{tabular}
\end{flushleft}

\paragraph{Details.}
\begin{enumerate}

\item \textbf{Asking the research question.}
The authors formulated the research question and the goal of establishing
a $\operatorname{poly}(\log n)$ versus $\operatorname{poly}(n)$
quantum-classical separation in the standard three-party NOF model.

\item \textbf{Coming up with the approach.}
During discussions with the authors on how to use their prior
work~\cite{yang2025deterministic,yang2026exponential,wang2026randomized}
to construct a three-party function along the lifting route, GPT-5.6 Pro suggested using the ABCD problem of Arunachalam, Girish, and
Lifshitz~\cite{arunachalam2023one} as the building block; following subsequent trials and manual screening, we formalized the
Interleaved Unitary Product problem (IUP)
(Definition~\ref{def:overview-abcde}).

\item \textbf{Proof development.} Focusing on the IUP problem, the authors proposed applying  
the structure-versus-pseudorandomness framework
to the analysis of cylinder intersections. 
Following this suggestion, GPT-5.6 Pro identified the regularity decomposition
of Abboud, Fischer, Kelley, Lovett, and Meka~\cite{abboud2024new}
and adapted it to the key regularity decomposition theorem
(Theorem~\ref{thm:regularity_decomposition}).
GPT-5.6 Pro developed the proof of the grid-norm bound
(Theorem~\ref{thm:Grid_norm}).
GPT-5.6 Pro also developed the regular and sparse estimates
(Lemmas~\ref{lem:regulary_corruption_bound}
and~\ref{lem:sparse_currption_bound}) and the resulting corruption bound
(Lemma~\ref{lem:corruption_bound}).

\item \textbf{Writing and exposition.}
The authors structured the manuscript and wrote a first draft.
AI tools (Codex GPT-6, Claude Opus 5.5) were subsequently used to
polish the language and presentation.

\item \textbf{Checking for bugs.}
AI helped check the proof and did not identify consequential errors in the proofs developed
along the successful route; it later helped simplify these proofs.

\item \textbf{Other supporting tasks.}
GPT-5.6 Pro also assisted with literature searches and language polishing.

\end{enumerate}
The authors reviewed and finalized all mathematical claims and proofs, and take full responsibility for the correctness, exposition, and attribution in the final manuscript.

\bibliographystyle{alpha}
\bibliography{references}

@book{knapp2016advanced,
  author    = {Anthony W. Knapp},
  title     = {Advanced Real Analysis},
  edition   = {Digital second},
  publisher = {Published by the author},
  year      = {2016},
  url       = {https://www.math.stonybrook.edu/~aknapp/download/a2-realanal-inside.pdf}
}

@misc{carlen2010divisibility,
  author       = {Eric A. Carlen},
  title        = {Notes on Divisibility of Non Atomic Measures for {Math 501, Fall 2010}},
  howpublished = {Lecture notes, Rutgers University},
  year         = {2010},
  month        = oct,
  url          = {https://sites.math.rutgers.edu/~carlen/501F10/atomic.pdf}
}

@inproceedings{wang2026randomized,
  title={Randomized and Quantum Lifting for One-Way Conservative {NOF} Model},
  author={Wang, Haoyu and Wu, Pei},
  booktitle={51st International Symposium on Mathematical Foundations of Computer Science (MFCS 2026)},
  pages={78:1--78:17},
  year={2026},
  organization={Schloss Dagstuhl--Leibniz-Zentrum f{\"u}r Informatik}
}

@inproceedings{kelley2023strong,
  title={Strong Bounds for 3-Progressions},
  author={Kelley, Zander and Meka, Raghu},
  booktitle={2023 IEEE 64th Annual Symposium on Foundations of Computer Science (FOCS)},
  pages={933--973},
  year={2023}
}

@article{klauck2018power,
  title={The power of one clean qubit in communication complexity},
  author={Klauck, Hartmut and Lim, Debbie},
  journal={arXiv preprint arXiv:1807.07762},
  year={2018}
}

@article{arunachalam2023one,
  title={One clean qubit suffices for quantum communication advantage},
  author={Arunachalam, Srinivasan and Girish, Uma and Lifshitz, Noam},
  journal={arXiv preprint arXiv:2310.02406},
  year={2023}
}

@inproceedings{abboud2024new,
  title={New graph decompositions and combinatorial {Boolean} matrix multiplication algorithms},
  author={Abboud, Amir and Fischer, Nick and Kelley, Zander and Lovett, Shachar and Meka, Raghu},
  booktitle={Proceedings of the 56th Annual ACM Symposium on Theory of Computing},
  pages={935--943},
  year={2024},
  note = {Full version: arXiv:2308.12451}
}

@inproceedings{kelley2024explicit,
  title={Explicit separations between randomized and deterministic number-on-forehead communication},
  author={Kelley, Zander and Lovett, Shachar and Meka, Raghu},
  booktitle={Proceedings of the 56th Annual ACM Symposium on Theory of Computing},
  pages={1299--1310},
  year={2024}
}

@inproceedings{kallaugher2019complexity,
  title={The complexity of counting cycles in the adjacency list streaming model},
  author={Kallaugher, John and McGregor, Andrew and Price, Eric and Vorotnikova, Sofya},
  booktitle={Proceedings of the 38th ACM SIGMOD-SIGACT-SIGAI Symposium on Principles of Database Systems},
  pages={119--133},
  year={2019}
}

@INPROCEEDINGS{HG90,
  author={Hastad, J. and Goldmann, M.},
  booktitle={Proceedings [1990] 31st Annual Symposium on Foundations of Computer Science}, 
  title={On the power of small-depth threshold circuits}, 
  year={1990},
  volume={},
  number={},
  pages={610-618 vol.2},
  doi={10.1109/FSCS.1990.89582}}

@inproceedings{chandra1983multi,
  title={Multi-party protocols},
  author={Chandra, Ashok K and Furst, Merrick L and Lipton, Richard J},
  booktitle={Proceedings of the fifteenth annual ACM symposium on Theory of computing},
  pages={94--99},
  year={1983}
}

@inproceedings{chakrabarti2007lower,
  title={Lower bounds for multi-player pointer jumping},
  author={Chakrabarti, Amit},
  booktitle={Twenty-Second Annual IEEE Conference on Computational Complexity (CCC'07)},
  pages={33--45},
  year={2007},
  organization={IEEE}
}

@INPROCEEDINGS {VW07,
author = {E. Viola and A. Wigderson},
booktitle = {2007 48th Annual IEEE Symposium on Foundations of Computer Science},
title = {One-Way Multi-Party Communication Lower Bound for Pointer Jumping with Applications},
year = {2007},
volume = {},
issn = {0272-5428},
pages = {427-437},
doi = {10.1109/FOCS.2007.34},
url = {https://doi.ieeecomputersociety.org/10.1109/FOCS.2007.34},
publisher = {IEEE Computer Society},
address = {Los Alamitos, CA, USA},
month = {oct}
}

@misc{jaber2025,
title={Quasipolynomial bounds for the corners theorem},
author={Michael Jaber and Yang P. Liu and Shachar Lovett and Anthony Ostuni and Mehtaab Sawhney},
year={2025},
eprint={2504.07006},
archivePrefix={arXiv},
primaryClass={math.CO},
url={https://arxiv.org/abs/2504.07006},
}

@article{babai1992multiparty,
  title={Multiparty protocols, pseudorandom generators for logspace, and time-space trade-offs},
  author={Babai, L{\'a}szl{\'o} and Nisant, Noam and Szegedy, M{\'a}ri{\'o}},
  journal={Journal of Computer and System Sciences},
  volume={45},
  number={2},
  pages={204--232},
  year={1992},
  publisher={Elsevier}
}

@inproceedings{ganor2014space,
  title={Space pseudorandom generators by communication complexity lower bounds},
  author={Ganor, Anat and Raz, Ran},
  booktitle={Approximation, Randomization, and Combinatorial Optimization. Algorithms and Techniques (APPROX/RANDOM 2014)},
  pages={692--703},
  year={2014},
  organization={Schloss Dagstuhl--Leibniz-Zentrum f{\"u}r Informatik}
}

@article{linial2017communication,
  title={On the communication complexity of high-dimensional permutations},
  author={Linial, Nati and Shraibman, Adi and others},
  journal={arXiv preprint arXiv:1706.02207},
  year={2017}
}

@inproceedings{drucker2014power,
  title={On the power of the congested clique model},
  author={Drucker, Andrew and Kuhn, Fabian and Oshman, Rotem},
  booktitle={Proceedings of the 2014 ACM symposium on Principles of distributed computing},
  pages={367--376},
  year={2014}
}

@article{beigel1994acc,
  title={On acc},
  author={Beigel, Richard and Tarui, Jun},
  journal={computational complexity},
  volume={4},
  number={4},
  pages={350--366},
  year={1994},
  publisher={Springer}
}

@inproceedings{bar2004exponential,
  title={Exponential separation of quantum and classical one-way communication complexity},
  author={Bar-Yossef, Ziv and Jayram, Thathachar S and Kerenidis, Iordanis},
  booktitle={Proceedings of the thirty-sixth annual ACM symposium on Theory of computing},
  pages={128--137},
  year={2004}
}

@inproceedings{buhrman1998quantum,
  title={Quantum vs. classical communication and computation},
  author={Buhrman, Harry and Cleve, Richard and Wigderson, Avi},
  booktitle={Proceedings of the thirtieth annual ACM symposium on Theory of computing},
  pages={63--68},
  year={1998}
}

@inproceedings{raz1999exponential,
  title={Exponential separation of quantum and classical communication complexity},
  author={Raz, Ran},
  booktitle={Proceedings of the thirty-first annual ACM symposium on Theory of computing},
  pages={358--367},
  year={1999}
}

@inproceedings{gavinsky2007exponential,
  title={Exponential separations for one-way quantum communication complexity, with applications to cryptography},
  author={Gavinsky, Dmitry and Kempe, Julia and Kerenidis, Iordanis and Raz, Ran and De Wolf, Ronald},
  booktitle={Proceedings of the thirty-ninth annual ACM symposium on Theory of computing},
  pages={516--525},
  year={2007}
}

@inproceedings{regev2011quantum,
  title={Quantum one-way communication can be exponentially stronger than classical communication},
  author={Regev, Oded and Klartag, Bo'az},
  booktitle={Proceedings of the forty-third annual ACM symposium on Theory of computing},
  pages={31--40},
  year={2011}
}

@inproceedings{gavinsky2016entangled,
  title={Entangled simultaneity versus classical interactivity in communication complexity},
  author={Gavinsky, Dmitry},
  booktitle={Proceedings of the forty-eighth annual ACM symposium on Theory of Computing},
  pages={877--884},
  year={2016}
}

@article{girish2022quantum,
  title={Quantum versus randomized communication complexity, with efficient players},
  author={Girish, Uma and Raz, Ran and Tal, Avishay},
  journal={computational complexity},
  volume={31},
  number={2},
  pages={17},
  year={2022},
  publisher={Springer}
}

@article{gavinsky2019quantum,
  title={Quantum versus classical simultaneity in communication complexity},
  author={Gavinsky, Dmitry},
  journal={IEEE Transactions on Information Theory},
  volume={65},
  number={10},
  pages={6466--6483},
  year={2019},
  publisher={IEEE}
}

@article{goos2024quantum,
  title={Quantum Communication Advantage in {TFNP}},
  author={G{\"o}{\"o}s, Mika and Gur, Tom and Jain, Siddhartha and Li, Jiawei},
  journal={arXiv preprint arXiv:2411.03296},
  year={2024}
}

@inproceedings{gavinsky2020bare,
  title={Bare quantum simultaneity versus classical interactivity in communication complexity},
  author={Gavinsky, Dmitry},
  booktitle={Proceedings of the 52nd Annual ACM SIGACT Symposium on Theory of Computing},
  pages={401--411},
  year={2020}
}

@inproceedings{ABK24,
  author       = {Scott Aaronson and
                  Harry Buhrman and
                  William Kretschmer},
  editor       = {Venkatesan Guruswami},
  title        = {A Qubit, a Coin, and an Advice String Walk into a Relational Problem},
  booktitle    = {Proceedings of the 15th Innovations in Theoretical Computer Science Conference},
  series       = {LIPIcs},
  volume       = {287},
  pages        = {1:1--1:24},
  publisher    = {Schloss Dagstuhl - Leibniz-Zentrum f{\"{u}}r Informatik},
  year         = {2024},
  url          = {https://doi.org/10.4230/LIPIcs.ITCS.2024.1},
  doi          = {10.4230/LIPICS.ITCS.2024.1}
}

@misc{KGDGGGHMNHA25,
  author       = {William Kretschmer and
                  Sabee Grewal and
                  Matthew DeCross and
                  Justin A. Gerber and
                  Kevin Gilmore and
                  Dan Gresh and
                  Nicholas Hunter-Jones and
                  Karl Mayer and
                  Brian Neyenhuis and
                  David Hayes and
                  Scott Aaronson},
  title        = {Demonstrating an Unconditional Separation between Quantum and Classical Information Resources}, 
  journal      = {CoRR},
  volume       = {abs/2509.07255},
  year         = {2025},
  url          = {https://arxiv.org/abs/2509.07255},
  eprinttype   = {arXiv},
  eprint       = {2509.07255},
  primaryClass = {quant-ph},
  howpublished = {arXiv:2509.07255}
}

@INPROCEEDINGS{QCCinNOF,
  author={Gavinsky, Dmitry and Pudlák, Pavel},
  booktitle={2008 23rd Annual IEEE Conference on Computational Complexity}, 
  title={Exponential Separation of Quantum and Classical Non-interactive Multi-party Communication Complexity}, 
  year={2008},
  volume={},
  number={},
  pages={332-339},
  doi={10.1109/CCC.2008.27}
}

@article{yang2025quantum,
  author       = {Guangxu Yang and
                  Jiapeng Zhang},
  title        = {Quantum versus Classical Separation in Simultaneous Number-on-Forehead Communication},
  journal      = {CoRR},
  volume       = {abs/2506.16804},
  year         = {2025},
  url          = {https://doi.org/10.48550/arXiv.2506.16804},
  doi          = {10.48550/ARXIV.2506.16804},
  eprinttype    = {arXiv},
  eprint       = {2506.16804}
}

@inproceedings{corrigan2019function,
  title={The function-inversion problem: Barriers and opportunities},
  author={Corrigan-Gibbs, Henry and Kogan, Dmitry},
  booktitle={Theory of Cryptography Conference},
  pages={393--421},
  year={2019},
  organization={Springer}
}

@inproceedings{li2020leakage,
  title={Leakage-resilient key exchange and two-seed extractors},
  author={Li, Xin and Ma, Fermi and Quach, Willy and Wichs, Daniel},
  booktitle={Annual International Cryptology Conference},
  pages={401--429},
  year={2020},
  organization={Springer}
}

@article{yang2025deterministic,
  title={Deterministic Lifting Theorems for One-Way Number-on-Forehead Communication},
  author={Yang, Guangxu and Zhang, Jiapeng},
  journal={arXiv preprint arXiv:2506.12420},
  year={2025}
}

@inproceedings{miles2013shielding,
  title={Shielding circuits with groups},
  author={Miles, Eric and Viola, Emanuele},
  booktitle={Proceedings of the forty-fifth annual ACM symposium on Theory of Computing},
  pages={251--260},
  year={2013}
}

@inproceedings{Yao,
  author    = {Andrew Chi-Chih Yao},
  title     = {Some Complexity Questions Related to Distributive Computing (Preliminary Report)},
  booktitle = {Proceedings of the Eleventh Annual ACM Symposium on Theory of Computing},
  pages     = {209--213},
  publisher = {ACM},
  year      = {1979},
  doi       = {10.1145/800135.804414}
}

@book{kushilevitz1997communication,
  author    = {Eyal Kushilevitz and Noam Nisan},
  title     = {Communication Complexity},
  publisher = {Cambridge University Press},
  year      = {1997},
  doi       = {10.1017/CBO9780511574948}
}

@book{nielsen2010quantum,
  author    = {Michael A. Nielsen and Isaac L. Chuang},
  title     = {Quantum Computation and Quantum Information},
  edition   = {10th Anniversary},
  publisher = {Cambridge University Press},
  year      = {2010},
  doi       = {10.1017/CBO9780511976667}
}

@article{trojek2005experimental,
  author  = {Pavel Trojek and Christian Schmid and Mohamed Bourennane and {\v{C}}aslav Brukner and Marek {\.{Z}}ukowski and Harald Weinfurter},
  title   = {Experimental quantum communication complexity},
  journal = {Physical Review A},
  volume  = {72},
  number  = {5},
  pages   = {050305},
  year    = {2005},
  doi     = {10.1103/PhysRevA.72.050305}
}

@article{xu2015experimental,
  author  = {Feihu Xu and Juan Miguel Arrazola and Kejin Wei and Wenyuan Wang and Pablo Palacios-Avila and Chen Feng and Shihan Sajeed and Norbert L{\"u}tkenhaus and Hoi-Kwong Lo},
  title   = {Experimental quantum fingerprinting with weak coherent pulses},
  journal = {Nature Communications},
  volume  = {6},
  pages   = {8735},
  year    = {2015},
  doi     = {10.1038/ncomms9735}
}

@article{guan2016observation,
  author  = {Jian-Yu Guan and Feihu Xu and Hua-Lei Yin and Yuan Li and Wei-Jun Zhang and Si-Jing Chen and Xiao-Yan Yang and Li Li and Li-Xing You and Teng-Yun Chen and Zhen Wang and Qiang Zhang and Jian-Wei Pan},
  title   = {Observation of Quantum Fingerprinting Beating the Classical Limit},
  journal = {Physical Review Letters},
  volume  = {116},
  number  = {24},
  pages   = {240502},
  year    = {2016},
  doi     = {10.1103/PhysRevLett.116.240502}
}

@article{kumar2019experimental,
  author  = {Niraj Kumar and Iordanis Kerenidis and Eleni Diamanti},
  title   = {Experimental demonstration of quantum advantage for one-way communication complexity surpassing best-known classical protocol},
  journal = {Nature Communications},
  volume  = {10},
  pages   = {4152},
  year    = {2019},
  doi     = {10.1038/s41467-019-12139-z}
}

@article{shen2025experimental,
  author  = {Ao Shen and Yu-Shuo Lu and Xiping Wu and Jinping Lin and Xiao-Yu Cao and Chengfang Ge and Shan-Feng Shao and Hua-Lei Yin and Lai Zhou and Zhiliang Yuan},
  title   = {Experimental Quantum Fingerprinting without the Shared Randomness Loophole},
  journal = {Physical Review Letters},
  volume  = {135},
  number  = {1},
  pages   = {010801},
  year    = {2025},
  doi     = {10.1103/v1fc-q1n9}
}

@article{knill1998power,
  author  = {Emanuel Knill and Raymond Laflamme},
  title   = {Power of One Bit of Quantum Information},
  journal = {Physical Review Letters},
  volume  = {81},
  number  = {25},
  pages   = {5672--5675},
  year    = {1998},
  doi     = {10.1103/PhysRevLett.81.5672}
}

@inproceedings{yang2026exponential,
  author        = {Guangxu Yang and Jiapeng Zhang},
  title         = {Exponential Separation of Quantum and Classical One-Way {Numbers-on-Forehead} Communication},
  booktitle     = {Proceedings of the IEEE Symposium on Foundations of Computer Science (FOCS)},
  year          = {2026},
  note          = {To appear},
  eprint        = {2603.22795},
  archivePrefix = {arXiv},
  primaryClass  = {quant-ph},
  url           = {https://arxiv.org/abs/2603.22795}
}

@article{gavinsky2026total,
  title={On the quantum communication complexity of total functions},
  author={Gavinsky, Dmytro},
  journal={arXiv preprint arXiv:2608.18784},
  year={2026}
}

\appendix
\section{Regularity decomposition for cylinder intersections}\label{app:proofs}

Abboud, Fischer, Kelley, Lovett, and Meka~\cite{abboud2024new} give a regularity decomposition of two Boolean matrices.\footnote{All theorem and algorithm numbers for Abboud et al.~\cite{abboud2024new}
refer to the full version.} We check that their recursion preserves the product at every $(X,Y,Z)$; multiplying by $\fun h(X,Y)$ then gives the required cylinder decomposition. We will restate all the cited theorems and algorithms  below for  completeness.

Throughout the argument, $\mathcal I,\mathcal J,\mathcal K$ are nonempty finite sets, $F:\mathcal J\times\mathcal K\to\{0,1\}$ and $G:\mathcal I\times\mathcal K\to\{0,1\}$. A matrix on a subdomain is extended by zero outside that subdomain; its density, grid norm, and min-degree use the uniform probability measure on its stated domain. We write $G[\mathcal I',\mathcal K']$ for a restriction. The notation $(X,Y,Z,A,B)$ of Abboud, Fischer, Kelley, Lovett, and Meka~\cite{abboud2024new} corresponds to $(\mathcal J,\mathcal K,\mathcal I,F,G^{\mathrm T})$ here.

The following GoodCube procedure selects $\mathcal I^*\times\mathcal K^*$ and splits $F$ on $\mathcal J\times\mathcal K^*$. 
\begin{lemma}[GoodCube, {\cite[Lemma 5.3]{abboud2024new}}]\label{lem:afklm_goodcube}
Let $0<\epsilon<1$, $0<\eta<1/2$, and let $d\ge1$ be an integer. If $\E G\ge2^{-d}$, there are nonempty $\mathcal I^*\subseteq\mathcal I$, $\mathcal K^*\subseteq\mathcal K$, and a list
$(\mathcal I_i,\mathcal J_i,\mathcal K_i,F_i)_{i=1}^{\ell}$, where $\mathcal I_i\subseteq\mathcal I^*$, $\mathcal J_i\subseteq\mathcal J$, $\mathcal K_i\subseteq\mathcal K^*$, and $F_i:\mathcal J_i\times\mathcal K_i\to\{0,1\}$, with the following properties. Let  $G_i=G[\mathcal I_i,\mathcal K_i]$.
\begin{enumerate}
\item $F(y,z)\mathbf1_{\mathcal K^*}(z)=\sum_{i=1}^{\ell}F_i(y,z)$ for every $(y,z)\in\mathcal J\times\mathcal K$.
\item For each $i$, either $\E F_i\le2^{-d}$, or both $F_i,G_i$ are $(\epsilon,d)$-regular and satisfy $\epsilon$-min-degree.
\item $\E G[\mathcal I^*,\mathcal K^*]\ge\E G$, and
\[
\sum_i|\mathcal I_i||\mathcal J_i||\mathcal K_i|\le(d+2)|\mathcal I^*||\mathcal J||\mathcal K^*|,
\]
\[
\sum_i|\mathcal I^*\setminus\mathcal I_i||\mathcal J_i||\mathcal K_i|\le\eta(d+2)|\mathcal I^*||\mathcal J||\mathcal K^*|.
\]
\item For a constant $c_\epsilon>0$ depending only on $\epsilon$,
\[
|\mathcal I^*||\mathcal K^*|\ge e^{-c_\epsilon d^4/\eta}|\mathcal I||\mathcal K|,
\qquad \ell\le e^{c_\epsilon d^3}.
\]
\end{enumerate}
The output is computed by a deterministic algorithm, denoted by $\fun{GoodCube}$.
\end{lemma}

The recursion below returns the pairs from $\fun{GoodCube}$, then recurses on $\mathcal I^*\setminus\mathcal I_i$ and on $G-G[\mathcal I^*,\mathcal K^*]$. It  follows Algorithm~5.5 of Abboud et al.~\cite{abboud2024new},
with the parameters and terminal step specified below.

\paragraph{$\fun{AB}$-decomposition.}
The $\fun{AB}$-decomposition procedure is as follows.

\noindent\textbf{Input:} $\mathcal I,\mathcal J,\mathcal K,F,G$, parameters $0<\epsilon<1$, integer $d\ge1$, and a counter $r\in\{0,\ldots,d\}$. The first call uses $r=0$; a call with an empty domain returns the empty list.

\noindent\textbf{Output:} A list of tuples $(\mathcal I_i,\mathcal J_i,\mathcal K_i,F_i,G_i)$.
\begin{enumerate}
\item If $\E G\le2^{-d}$, return $\{(\mathcal I,\mathcal J,\mathcal K,F,G)\}$.
\item If $r=d$, return the following list and stop.
\begin{enumerate}
\item If $|\mathcal I||\mathcal K|\ge2^d$, partition the $1$ entries of $G$ into groups of size $\lfloor2^{-d}|\mathcal I||\mathcal K|\rfloor$, except possibly the last group. For each group's indicator $G^{(j)}$, return $(\mathcal I,\mathcal J,\mathcal K,F,G^{(j)})$.
\item Otherwise, for each $G(x,z)=1$, put $\mathcal J_{x,z}=\{y:F(y,z)=1\}$. If this set is nonempty, return
\[
\bigl(\{x\},\mathcal J_{x,z},\{z\},F[\mathcal J_{x,z},\{z\}],G[\{x\},\{z\}]\bigr).
\]
\end{enumerate}
\item Set $\eta=1/[4(d+2)^2]$. Since $\E G>2^{-d}$, apply the GoodCube procedure of Lemma~\ref{lem:afklm_goodcube} to $(\mathcal I,\mathcal J,\mathcal K,F,G;\epsilon,\eta,d)$. For its output, put
\[
G_i=G[\mathcal I_i,\mathcal K_i],\qquad
G_i'=G[\mathcal I^*\setminus\mathcal I_i,\mathcal K_i],\qquad
G^*=G[\mathcal I^*,\mathcal K^*].
\]
\item Make the recursive calls
\[
\mathcal O_i=\text{$\fun{AB}$-decomposition}(\mathcal I^*\setminus\mathcal I_i,\mathcal J_i,\mathcal K_i,F_i,G_i';\epsilon,d,r+1)
\quad (i=1,\ldots,\ell),
\]
\[
\mathcal O^*=\text{$\fun{AB}$-decomposition}(\mathcal I,\mathcal J,\mathcal K,F,G-G^*;\epsilon,d,r).
\]
\item Return the list $(\mathcal I_i,\mathcal J_i,\mathcal K_i,F_i,G_i)_{i=1}^{\ell}$ followed by $\mathcal O_1,\ldots,\mathcal O_\ell,\mathcal O^*$.
\end{enumerate}

Step~2(a) returns at most $2^{d+1}$ tuples, each with $\E G^{(j)}\le2^{-d}$; both output matrices in step~2(b) are identically $1$. Fixing orders on the finite sets makes the choices deterministic.

The small omitted sets bound the total size of the returned domains, and the density removed at each call bounds the number of calls. The following is the bound of Abboud et al.~\cite[Theorem 3.1]{abboud2024new} with the stated terminal step.
\begin{lemma}[$\fun{AB}$-decomposition bounds]\label{lem:afklm_ab_bounds}
The first call of $\fun{AB}$-decomposition returns $k$ tuples satisfying
\[
\sum_{i=1}^{k}|\mathcal I_i||\mathcal J_i||\mathcal K_i|\le4(d+2)^2|\mathcal I||\mathcal J||\mathcal K|,
\qquad k\le\exp(c_\epsilon d^7).
\]
Every output is sparse in at least one matrix, with density at most $2^{-d}$, or both matrices are $(\epsilon,d)$-regular and satisfy $\epsilon$-min-degree.
\end{lemma}

\begin{proof}[Proof of Lemma~\ref{lem:afklm_ab_bounds}]
\leavevmode
\paragraph{Step 1. Bound the total size of the returned domains.}
Write $V=|\mathcal I||\mathcal J||\mathcal K|$ and $\delta=\E G$ for the current call. The sum $\sum_i|\mathcal I_i||\mathcal J_i||\mathcal K_i|$ over its returned tuples is at most $V$ if $\delta\le2^{-d}$, and otherwise at most
\begin{equation}\label{eq:ab_volume_potential}
2^{r+2}(d+2)(d+2+\log_2\delta)V.
\end{equation}
Induct lexicographically on $(d-r,|\{(x,z):G(x,z)=1\}|)$. The lists returned in steps~1 and~2 of $\fun{AB}$-decomposition satisfy $\sum_i|\mathcal I_i||\mathcal J_i||\mathcal K_i|\le2^{d+1}V$. In a recursive call put $V^*=|\mathcal I^*||\mathcal J||\mathcal K^*|$. The GoodCube size bounds in Lemma~\ref{lem:afklm_goodcube}(3) are
\[
\sum_i|\mathcal I_i||\mathcal J_i||\mathcal K_i|\le(d+2)V^*,\qquad
\sum_i|\mathcal I^*\setminus\mathcal I_i||\mathcal J_i||\mathcal K_i|\le\eta(d+2)V^*.
\]
For each recursive call at $r+1$, the induction hypothesis bounds the total size of its returned domains by $2^{r+3}(d+2)^2$ times the product of the sizes of its three input sets. Hence the pairs supplied by $\fun{GoodCube}$ and the outputs in $\mathcal O_1,\ldots,\mathcal O_\ell$ contribute at most
\[
(d+2)V^*+2^{r+3}(d+2)^2\eta(d+2)V^*\le2^{r+2}(d+2)V^*.
\]
GoodCube also guarantees $\E G[\mathcal I^*,\mathcal K^*]\ge\delta$. Thus the remaining density satisfies
\[
\delta^*:=\E(G-G^*)=\delta-\frac{V^*}{V}\E G[\mathcal I^*,\mathcal K^*]\le\delta(1-V^*/V).
\]
If $\delta^*>2^{-d}$, including the outputs in $\mathcal O^*$ gives \eqref{eq:ab_volume_potential}, since
\[
\log_2\delta^*\le\log_2\delta-V^*/V
\quad\Longrightarrow\quad
V^*+(d+2+\log_2\delta^*)V\le(d+2+\log_2\delta)V.
\] If $\delta^*\le2^{-d}$, the sum over all returned tuples is at most $2^{r+2}(d+2)V^*+V$, which is also at most \eqref{eq:ab_volume_potential} since $d+2+\log_2\delta>2$. Taking $r=0$ proves the claimed bound on this sum.

\paragraph{Step 2. Bound the number of outputs.}
The bounds on $|\mathcal I^*||\mathcal K^*|$ and $\ell$ in Lemma~\ref{lem:afklm_goodcube}(4), with $\eta=1/[4(d+2)^2]$, give numbers $M=\exp(O_\epsilon(d^6))\ge1$ and $L=\exp(O_\epsilon(d^3))\ge1$ such that
\[
\frac{V^*}{V}\ge\frac1M,\qquad \ell\le L,\qquad
\delta-\delta^*\ge\delta\frac{V^*}{V}>\frac{2^{-d}}M.
\]
Let $k_r$ be the maximum number of tuples returned by a call with counter $r$. Each chain of calls on $G-G^*$ has at most $2^dM$ nonterminal calls, so
\[
k_d\le2^{d+1},\qquad k_r\le2^dML(1+k_{r+1})+2^{d+1}\quad(0\le r<d).
\]
Induction on $d-r$ gives
\[
k\le k_0\le\bigl(2^{d+2}M(L+1)\bigr)^{d+1}=\exp(O_\epsilon(d^7)).
\]
In $\fun{AB}$-decomposition, steps~1 and~2(a) return pairs with a matrix of density at most $2^{-d}$, while step~2(b) returns two all-$1$ matrices. The pairs supplied by $\fun{GoodCube}$ are sparse or regular by Lemma~\ref{lem:afklm_goodcube}(2). The bounds on recursive calls also show that $\fun{AB}$-decomposition terminates.
\end{proof}

We now prove the regularity decomposition on cylinder intersections.
\begin{proof}[Proof of Theorem~\ref{thm:regularity_decomposition}]\leavevmode
\paragraph{Step 1. Preserve $F(y,z)G(x,z)$ at every triple.}
GoodCube preserves $F$ on $\mathcal J\times\mathcal K^*$ by Lemma~\ref{lem:afklm_goodcube}(1):
\[
\sum_iF_i(y,z)=F(y,z)\mathbf1_{\mathcal K^*}(z),\qquad
F_i(y,z)\bigl(G_i(x,z)+G_i'(x,z)\bigr)=F_i(y,z)G(x,z)\mathbf1_{\mathcal I^*}(x).
\]
The second identity follows from $\mathcal I^*=\mathcal I_i\sqcup(\mathcal I^*\setminus\mathcal I_i)$. Summing it and adding $F(G-G^*)$ gives
\begin{equation}\label{eq:ab_triple_split}
F(y,z)G(x,z)=\sum_iF_i(y,z)G_i(x,z)+\sum_iF_i(y,z)G_i'(x,z)+F(y,z)(G-G^*)(x,z).
\end{equation}
The last two terms are the recursive products in step~4 of $\fun{AB}$-decomposition. In step~2(a) of $\fun{AB}$-decomposition, the indicators $G^{(j)}$ sum to $G$; in step~2(b), each triple $(x,y,z)$ with $F(y,z)G(x,z)=1$ appears only in the output indexed by $(x,z)$. Induction therefore gives
\begin{equation}\label{eq:ab_pointwise_identity}
F(y,z)G(x,z)=\sum_{i=1}^{k}F_i(y,z)G_i(x,z).
\end{equation}

\paragraph{Step 2. Return functions on $\mathcal X,\mathcal Y,\mathcal Z$.}
Let $\mathcal I,\mathcal J,\mathcal K$ index the equal-measure partitions defining the finite-type functions $\fun f,\fun g,\fun h$. Write $x,y,z$ for the labels of the sets containing $X,Y,Z$, and define the Boolean matrices
\[
F(y,z)=\fun f(Y,Z),\qquad G(x,z)=\fun g(X,Z),\qquad H(x,y)=\fun h(X,Y).
\]
Apply $\fun{AB}$-decomposition with the integer $q=d+1$ in place of $d$ and $\min\{\epsilon,1/2\}\in(0,1)$ in place of $\epsilon$. For each returned tuple, let $\mathcal X_i,\mathcal Y_i,\mathcal Z_i$ be the unions of the sets indexed by $\mathcal I_i,\mathcal J_i,\mathcal K_i$, respectively, and set
\[
\fun f_i(Y,Z)=F_i(y,z),\qquad \fun g_i(X,Z)=G_i(x,z),\qquad \mathcal B_i=\mathcal X_i\times\mathcal Y_i\times\mathcal Z_i.
\]
Their supports have the required form. Multiplying \eqref{eq:ab_pointwise_identity} by $H(x,y)$ gives
\[
\sum_i\fun f_i(Y,Z)\fun g_i(X,Z)\fun h(X,Y)=F(y,z)G(x,z)H(x,y)=\fun c(X,Y,Z).
\]
All terms are Boolean, so they are also pairwise disjoint wherever they equal $1$.

Equal cell measures identify the normalized integrals with uniform finite averages. Hence the densities, grid norms, and averages defining min-degree agree with those of $F_i,G_i$. Lemma~\ref{lem:afklm_ab_bounds} classifies each output: either a density is at most $2^{-q}<2^{-d}$, or both functions are regular and satisfy min-degree. In the latter case, monotonicity of probability-space $L^p$ norms gives
\[
\|\fun f_i\|_{U(2,d),i}\le\|\fun f_i\|_{U(2,q),i}\le(1+\epsilon)\delta_{\fun f,i},\qquad
\|\fun g_i\|_{U(2,d),i}\le\|\fun g_i\|_{U(2,q),i}\le(1+\epsilon)\delta_{\fun g,i},
\]
and increasing the min-degree parameter to $\epsilon$ only weakens its lower bound.

The same lemma bounds the finite total size and number of outputs, which translate to
\[
\sum_i\mu_0(\mathcal B_i)=\frac{\sum_i|\mathcal I_i||\mathcal J_i||\mathcal K_i|}{|\mathcal I||\mathcal J||\mathcal K|}\le4(d+3)^2=O(d^2),
\qquad k\le\exp(c_\epsilon d^7),
\]
after increasing $c_\epsilon$. This proves all three claims.
\end{proof}

\section{The grid norm bound}\label{app:agl}

Arunachalam, Girish, and Lifshitz~\cite{arunachalam2023one} bound how much the constraint $ABCD=I_m$ can increase the average of $\fun u(A,C)\fun w(B,D)$. Applying this estimate to the two matrix products in $W=PE_1QE_2$ will give the grid norm bound.

\begin{theorem}[{\cite[Lemma~5.2]{arunachalam2023one}}]\label{thm:AGL_product}
There is a universal constant $\kappa>0$ such that the following holds. Let $A,B,C,D$ be independent Haar matrices in $\sugroup(m)$. For measurable Boolean functions $\fun u,\fun w:\sugroup(m)^2\to\{0,1\}$ with $\E[\fun u],\E[\fun w]\ge e^{-\kappa\sqrt m}$, we have
\[
\Big|\E[\fun u(A,C)  \fun w(B, (ABC)^{-1}) ] -\E[\fun u(A,C) \fun w(B,D)] \Big| \le \frac{\E[\fun u]\E[\fun w]}{30}.
\]
\end{theorem}

Write $\mathcal G=\sugroup(m)$ and $\mathcal H=\mathcal G^2$, with coordinatewise multiplication. Let $\kappa_0>0$ be the constant in Theorem~\ref{thm:AGL_product}, and put $\lambda_0=31/30$ and $\lambda=2\lambda_0$.
In either group, conditioning on $ABCD=T$ means sampling $A,B,C$ independently by Haar measure and setting $D=(ABC)^{-1}T$. Unconditioned averages below use independent Haar pairs.

The next lemma removes the density condition by bounding small averages directly, then applies the resulting estimate twice to handle both matrix products.
\begin{lemma}[Weighted product estimates]\label{lem:AGL_weighted_products}
\begin{enumerate}
\item\label{enu:calG} For $T\in\mathcal G$ and measurable $\fun a,\fun b:\mathcal G^2\to[0,1]$,
\begin{equation}\label{eq:AGL_weighted_one}
\E[\fun a(A,C)\fun b(B,D)\mid ABCD=T]
\le \lambda_0\E[\fun a]\E[\fun b]+e^{-\kappa_0\sqrt m}.
\end{equation}
\item For $L\in\mathcal H$ and measurable $\fun a,\fun b:\mathcal H^2\to[0,1]$,
\begin{equation}\label{eq:AGL_weighted_two}
\E[\fun a(A,C)\fun b(B,D)\mid ABCD=L]
\le\lambda_0^2\E[\fun a]\E[\fun b]+(1+\lambda_0)e^{-\kappa_0\sqrt m}.
\end{equation}
\end{enumerate}
\end{lemma}

\begin{proof}
For~(1), first let $\fun a,\fun b$ be Boolean. If both averages are at least $e^{-\kappa_0\sqrt m}$, apply the relative-error bound in Theorem~\ref{thm:AGL_product} to $\fun a$ and $(B,D)\mapsto\fun b(B,DT)$. Haar invariance preserves $\E[\fun b]$, giving
\[
\E[\fun a(A,C)\fun b(B,D)\mid ABCD=T]
\le\lambda_0\E[\fun a]\E[\fun b].
\]
Otherwise, both $(A,C)$ and $(B,D)$ still have product Haar marginals, so
\[
\E[\fun a\fun b\mid ABCD=T]
\le\min\{\E[\fun a],\E[\fun b]\}<e^{-\kappa_0\sqrt m}.
\]
For functions in $[0,1]$, integrate the Boolean inequality using
$\fun a=\int_0^1\mathbf1[\fun a\ge s]\,ds$ and
$\fun b=\int_0^1\mathbf1[\fun b\ge t]\,dt$.

For~(2), write $A=(A_1,A_2)$, and similarly for $B,C,D,L$. Fix the second coordinates and apply~(1) to the first coordinates. The resulting functions
$\fun a_2(A_2,C_2)=\E_{A_1,C_1}[\fun a(A,C)]$ and
$\fun b_2(B_2,D_2)=\E_{B_1,D_1}[\fun b(B,D)]$ lie in $[0,1]$.
Averaging over the second coordinates and applying~(1) again gives
\[
\E[\fun a\fun b\mid ABCD=L]
\le\lambda_0\E[\fun a_2\fun b_2\mid A_2B_2C_2D_2=L_2]+e^{-\kappa_0\sqrt m}
\le\lambda_0^2\E[\fun a]\E[\fun b]+(1+\lambda_0)e^{-\kappa_0\sqrt m}.\qedhere
\]
\end{proof}

We now prove Theorem~\ref{thm:Grid_norm} using Lemma~\ref{lem:AGL_weighted_products} in the following section.

\subsection{Proof of Theorem~\ref{thm:Grid_norm}}\label{app:proof-grid-norm}

\paragraph{Sampling $\mu_1$ through $P,Q$.}
For the proof below, write $A=(A_1,A_2)$ and similarly for $B,C,D$.
Equivalently, sample $\mu_1$ by first sampling independent Haar matrices $P,Q$ and putting $L=(P,Q)$. Given $L$, sample $E_1$ by Haar measure and set $E_2=(PE_1Q)^{-1}$; independently, sample $A,B,C$ by product Haar measure and set $D=(ABC)^{-1}L$.
Thus $(X,Y)$ and $Z$ are independent conditional on $L$, with
\[
ABCD=L,\qquad PE_1QE_2=I_m.
\]
Unconditionally, $L\sim\nu^{\otimes2}$ and $Z\sim\nu^{\otimes2}$.
Lemma~\ref{lem:AGL_weighted_products}(1) therefore controls functions of $(P,Q)$ and $(E_1,E_2)$, while part~(2) controls functions of
$X=(A_1,C_1,A_2,C_2)$ and $Y=(B_1,D_1,B_2,D_2)$ conditional on $L$.

We expand the $d$-th moment using conditional samples of $Z$. The product bounds above replace the dependent averages by independent ones, whose second moments give the grid norms.
Fix $\fun u:\mathcal G^2\to[0,1]$.

\paragraph{Step 1. Bound $\E_L[(\E[\fun u(Z)\mid L])^d]$.}
Define
$\fun r(L)=\E[\fun u(Z)\mid L]$ and
$\mathcal S=\{L:\fun r(L)>2\lambda_0\E[\fun u]\}$.
Lemma~\ref{lem:AGL_weighted_products}(1) controls functions of the alternating factors in $PE_1QE_2=I_m$. Apply it with
$(A,B,C,D)=(P,E_1,Q,E_2)$, $T=I_m$, and
$(\fun a,\fun b)=(\mathbf1_{\mathcal S},\fun u)$.
Both functions lie in $[0,1]$, and the displayed product condition holds;
hence
\[
\E[\fun r\mathbf1_{\mathcal S}]
\le\lambda_0\E[\fun u]\Pr[L\in\mathcal S]+e^{-\kappa_0\sqrt m}
\le\tfrac12\E[\fun r\mathbf1_{\mathcal S}]+e^{-\kappa_0\sqrt m}.
\]
Since $0\le\fun r\le1$, for every integer $d\ge1$ we obtain
\begin{equation}\label{eq:AGL_outer_moment}
\E_L\!\left[\big(\E[\fun u(Z)\mid L]\big)^d\right]
\le \min\{1,(2\lambda_0\E[\fun u])^d\}+2e^{-\kappa_0\sqrt m}.
\end{equation}

\paragraph{Step 2. Bound the moments of the products of $\fun f$ and $\fun g$.}
Given $L$, let $Z^{(1)},\ldots,Z^{(d)}$ be independent samples from
the conditional distribution of $Z$; write their tuple as $Z^{[d]}$.
For fixed $L,Z^{[d]}$, Lemma~\ref{lem:AGL_weighted_products}(2) replaces the average of a function of $X$ times a function of $Y$ by their independent averages. Apply it to
$\prod_{j=1}^d\fun g(X,Z^{(j)})$ and
$\prod_{j=1}^d\fun f(Y,Z^{(j)})$.
These are Boolean functions of $(A,C)$ and $(B,D)$, respectively.
Given $L$, the samples $Z^{(1)},\ldots,Z^{(d)}$ are mutually independent and independent of $(X,Y)$. Thus
\begin{align}
&\big\|\E_{\mu_1}[\fun f\fun g\mid X,Y]\big\|_d^d\notag\\
&=\E_{L,Z^{[d]}}\!\left[
 \E\!\left[\prod_{j=1}^d\fun f(Y,Z^{(j)})\fun g(X,Z^{(j)})
 \mid L,Z^{[d]}\right]\right]\notag\\
&\le\lambda_0^2\E_{L,Z^{[d]}}\!\left[
 \left(\E_Y\!\left[\prod_{j=1}^d\fun f(Y,Z^{(j)})\right]\right)
 \left(\E_X\!\left[\prod_{j=1}^d\fun g(X,Z^{(j)})\right]\right)
 \right]+(1+\lambda_0)e^{-\kappa_0\sqrt m}.
 \label{eq:AGL_inner_moment}
\end{align}
Here the inner averages on the last line use $X\sim\nu^{\otimes4}$ and $Y\sim\nu^{\otimes4}$ independently.
For independent $Y,Y'\sim\nu^{\otimes4}$, apply the conditional-moment bound~\eqref{eq:AGL_outer_moment}
to $\fun u(Z)=\fun f(Y,Z)\fun f(Y',Z)$ and average over $Y,Y'$:
\begin{align}
\E_{L,Z^{[d]}}\!\left[
 \left(\E_Y\!\left[\prod_{j=1}^d\fun f(Y,Z^{(j)})\right]\right)^2\right]\notag
&=\E_{Y,Y',L}\!\left[
 \big(\E[\fun f(Y,Z)\fun f(Y',Z)\mid L,Y,Y']\big)^d\right]\notag\\
&\le\min\{1,(2\lambda_0)^d\|\fun f\|_{U(2,d)}^{2d}\}
   +2e^{-\kappa_0\sqrt m}.
\label{eq:AGL_square_moment}
\end{align}
The last line uses
$\E[\min\{1,V\}]\le\min\{1,\E[V]\}$ for a nonnegative random variable $V$, and the definition of
$\|\cdot\|_{U(2,d)}$; the same bound holds for $\fun g$.

\paragraph{Step 3. Bound $\|\E_{\mu_1}[\fun f\fun g\mid X,Y]\|_d$.}
For $0\le a,b\le1$ and $q\ge0$,
$\sqrt{(a+q)(b+q)}\le\sqrt{ab}+2\sqrt q+q$.
The expectation of the product in~\eqref{eq:AGL_inner_moment} is at most the geometric mean of the two second moments, by Cauchy--Schwarz. Substitute their bounds~\eqref{eq:AGL_square_moment} and the preceding inequality with
$q=2e^{-\kappa_0\sqrt m}$.  Since
$2\sqrt2\lambda_0^2+2\lambda_0^2+1+\lambda_0<8$, we get
\[
\big\|\E_{\mu_1}[\fun f\fun g\mid X,Y]\big\|_d^d
\le\lambda_0^2(2\lambda_0)^d
 \big(\|\fun f\|_{U(2,d)}\|\fun g\|_{U(2,d)}\big)^d
 +8e^{-\kappa_0\sqrt m/2}.
\]
Recall $\lambda=2\lambda_0$.  Taking the $d$th root gives
\begin{equation}\label{eq:AGL_root_bound}
\big\|\E_{\mu_1}[\fun f\fun g\mid X,Y]\big\|_d
\le\lambda^{1+2/d}\|\fun f\|_{U(2,d)}\|\fun g\|_{U(2,d)}
 +(8e^{-\kappa_0\sqrt m/2})^{1/d}.
\end{equation}
If $\sqrt m\ge4\log(8)/\kappa_0$, then
\[
(8e^{-\kappa_0\sqrt m/2})^{1/d}\le e^{-\kappa_0\sqrt m/(4d)}.
\]
For the remaining dimensions, put
$p=\|\fun f\|_{U(2,d)}\|\fun g\|_{U(2,d)}$.
By Claim~\ref{claim:haar_measure_marginal}, $\mu_1(\fun f)=\mu_0(\fun f)$ and $\mu_1(\fun g)=\mu_0(\fun g)$. Since $\fun f\fun g\in\{0,1\}$ and each grid norm dominates the corresponding average, Jensen's inequality gives
\[
\big\|\E_{\mu_1}[\fun f\fun g\mid X,Y]\big\|_d^d
\le\min\{\mu_0(\fun f),\mu_0(\fun g)\}\le\sqrt p.
\]
If $p\ge1/\lambda$, the main term in the claimed bound is at least $1$.
Otherwise, the last display gives
$\|\E_{\mu_1}[\fun f\fun g\mid X,Y]\|_d\le p^{1/(2d)}<\lambda^{-1/(2d)}$.
Choose
$0<\kappa\le\min\{\kappa_0/4,\kappa_0\log(\lambda)/(8\log(8))\}$;
then $\lambda^{-1/(2d)}\le e^{-\kappa\sqrt m/d}$ in these remaining
dimensions.  This proves the stated bound for every $d\ge2$.

\section{Additional Proofs}\label{app:additional-proofs}

\begin{proof}[Proof of Claim~\ref{claim:unitarization}]
The singular values of $U$ lie in $\big[\sqrt{1-\sigma},\sqrt{1+\sigma} \big]$, since
$\|U^\dagger U-I_m\|_{\operatorname{op}}\le\|U^\dagger U-I_m\|_{\mathrm F}\le\sigma<1$.
Thus $U$ is invertible. Let $\mathcal{S}_{U}$ be the set of  singular values of $U$; then its polar factor satisfies
\[
\widetilde U^\dagger\widetilde U=I_m,\qquad
\|\widetilde U-U\|_{\operatorname{op}}
=\max_{s\in \mathcal{S}_{U} }\frac{|s^2-1|}{s+1}\le\sigma,
\qquad \|U\|_{\operatorname{op}}\le\sqrt{1+\sigma}.
\]
Write $U_1,\ldots,U_{10}$ for the factors in $W$, in their product order. Replacing one factor at a time gives
\[
\widetilde W-W
=\sum_{r=1}^{10}\widetilde U_1\cdots\widetilde U_{r-1}
(\widetilde U_r-U_r)U_{r+1}\cdots U_{10}.
\]
The preceding bounds and $|\trace M|\le m\|M\|_{\operatorname{op}}$ now give
\[
\frac{|\operatorname{Re}\trace\widetilde W-\operatorname{Re}\trace W|}{m}
\le\|\widetilde W-W\|_{\operatorname{op}}
\le\sum_{r=1}^{10}\sigma(1+\sigma)^{(10-r)/2}
\le10\sigma(1+\sigma)^{9/2}.\qedhere
\]
\end{proof}

\begin{proof}[Proof of Claim~\ref{claim:haar_measure_marginal}]
Let $U_1,\ldots,U_{10}$ be the factors of $W$ in order. Under $\mu_1$,
\[
(U_1,\ldots,U_9)\sim\nu^{\otimes9},
\qquad
U_{10}=(U_1\cdots U_9)^{-1}.
\]
Fix $j\in[9]$ and condition on $(U_r)_{r\in[9]\setminus\{j\}}$.
Then $U_j$ remains Haar, and
\[
U_{10}=(U_{j+1}\cdots U_9)^{-1}U_j^{-1}(U_1\cdots U_{j-1})^{-1}.
\]
Both surrounding products are fixed. Thus Haar invariance gives, for every measurable $S\subseteq\sugroup(m)$,
\[
\Pr_{\mu_1}\!\left[U_{10}\in S\mid(U_r)_{r\in[9]\setminus\{j\}}\right]=\nu(S).
\]
Together with the independent Haar distribution of the other eight matrices, this implies
\[
(U_r)_{r\in[10]\setminus\{j\}}\sim\nu^{\otimes9}
\qquad(j\in[10]).
\]
The views $(X,Y)$, $(X,Z)$, and $(Y,Z)$ contain $8$, $6$, and $6$ distinct matrices, respectively, so each has the same product Haar distribution as under $\mu_0$.
\end{proof}

\begin{proof}[Proof of Claim~\ref{claim:rounding}]
Let $\widetilde U_{\mathrm r}$ denote the rounded version of a Haar matrix $U$; the subscript $\mathrm r$ marks rounding, as opposed to the polar factor $\widetilde U$. Put $\Delta=\widetilde U_{\mathrm r}-U$.
Each real coordinate changes by at most $2^{-b-1}$, so
\[
\|\Delta\|_{\mathrm F}\le\frac{m2^{-b}}{\sqrt2}\le\frac mn,
\qquad
\|\widetilde U_{\mathrm r}^\dagger\widetilde U_{\mathrm r}-I_m\|_{\mathrm F}
\le2\|\Delta\|_{\mathrm F}+\|\Delta\|_{\mathrm F}^2
\le\frac{2m}{n}+\frac{m^2}{n^2}<10^{-4}
\]
for all sufficiently large $n$, because $m=\Theta(\sqrt{n/\log n})$.
Let $\widetilde W_{\mathrm r}$ be the product of the ten rounded factors. Expanding the product difference one factor at a time, as in Claim~\ref{claim:unitarization}, gives
\[
\frac{|\operatorname{Re}\trace\widetilde W_{\mathrm r}-\operatorname{Re}\trace W|}{m}
\le\|\widetilde W_{\mathrm r}-W\|_{\operatorname{op}}
\le\frac{10m}{n}\left(1+\frac mn\right)^9<10^{-3}.
\]
Under $\mu_1$, $W=I_m$, so
\[
\frac{\operatorname{Re}\trace W}{m}=1,\qquad
\frac{\operatorname{Re}\trace\widetilde W_{\mathrm r}}{m}>1-10^{-3}>0.9,\qquad
\fun F_n(\fun R(X,Y,Z))=1.
\]
Under $\mu_0$, $W$ is Haar. Haar invariance gives
\[
\E_{\mu_0}[W_{ii}\overline{W_{jj}}]=\frac{\mathbf1[i=j]}m,
\qquad \E_{\mu_0}[|\trace W|^2]=1.
\]
Markov's inequality gives the continuous-input estimate
\[
\Pr_{\mu_0}[\operatorname{Re}\trace W>0.1m]
\le\frac{\E_{\mu_0}[|\trace W|^2]}{0.1^2m^2}
=\frac{100}{m^2}\le\frac1m.
\]
The rounded matrices satisfy 
$\bigl\|\widetilde U_{\mathrm r}^{\dagger}\widetilde U_{\mathrm r}-I_m\bigr\|_{\mathrm F}\le10^{-4},$  and rounding changes the normalized trace by less than $10^{-3}$. Hence
\[
\Pr_{\widetilde\mu_0}[\fun F_n\ne0]
\le\Pr_{\mu_0}[\operatorname{Re}\trace W>0.099m]
\le\frac{\E_{\mu_0}[|\trace W|^2]}{0.099^2m^2}
\le\frac1m
\]
for all sufficiently large $m$. The event $\fun F_n\ne0$ includes inputs outside the promise, so this proves the stated probability of $\fun F_n=0$.
\end{proof}

\begin{proof}[Proof of Claim~\ref{claim:step_function_arppox}]
All pairwise expectations below use independent Haar inputs.

\paragraph{Step 1. Construct Boolean approximations on finite partitions.}
Let $\mathcal E=\{\fun f=1\}$.
By compact inner and open outer approximation for Haar measure
(see Knapp~\cite[Section~VI.2]{knapp2016advanced}),
there are a compact set $\mathcal K$ and an open set $\mathcal O$ with
\[
\mathcal K\subseteq\mathcal E\subseteq\mathcal O,
\qquad
(\nu^{\otimes4}\otimes\nu^{\otimes2})(\mathcal O\setminus\mathcal K)
<\frac1{6r}.
\]
Every point of $\mathcal K$ has an open rectangle neighborhood contained
in $\mathcal O$. A finite subcover gives a union of rectangles $\mathcal D$
with $\mathcal K\subseteq\mathcal D\subseteq\mathcal O$.
Thus $\fun f':=\mathbf1_{\mathcal D}$ satisfies
\[
\E_{Y,Z}|\fun f-\fun f'|
=(\nu^{\otimes4}\otimes\nu^{\otimes2})(\mathcal E\triangle\mathcal D)
\le(\nu^{\otimes4}\otimes\nu^{\otimes2})(\mathcal O\setminus\mathcal K)
<\frac1{6r}.
\]
Construct $\fun g',\fun h'$ similarly.
Partition each of $\mathcal X,\mathcal Y,\mathcal Z$ according to membership
in all rectangle sides used in these three constructions.
The resulting finite partitions $\mathcal P_X,\mathcal P_Y,\mathcal P_Z$
make $\fun f',\fun g',\fun h'$ constant on the corresponding cell products,
with each approximation error at most $1/(6r)$.

\paragraph{Step 2. Make the partitions equal in measure.}
Let $s=\max\{|\mathcal P_X|,|\mathcal P_Y|,|\mathcal P_Z|\}$ and choose
an integer $q\ge12rs$.
Since Haar measure is nonatomic for $m\ge2$, every measurable set
contains subsets of any prescribed measure between zero and its own
measure (see Carlen~\cite[p.~2]{carlen2010divisibility}).
Thus each $S\in\mathcal P_X$ splits into
$\lfloor q\nu^{\otimes4}(S)\rfloor$ sets of measure $1/q$
and a remainder $R_S$ of measure less than $1/q$.
Their union $R_X:=\bigcup_{S\in\mathcal P_X}R_S$ satisfies
\[
\nu^{\otimes4}(R_X)<\frac{s}{q},
\qquad
q\nu^{\otimes4}(R_X)
=q-\sum_{S\in\mathcal P_X}\lfloor q\nu^{\otimes4}(S)\rfloor
\in\mathbb Z_{\ge0}.
\]
Therefore $R_X$ also splits into sets of measure $1/q$,
completing an equal-measure partition up to null sets.
Repeat this on $\mathcal Y,\mathcal Z$, obtaining remainders with
$\nu^{\otimes4}(R_Y),\nu^{\otimes2}(R_Z)<s/q$.
Define
\begin{align*}
\fun f_r=\fun f'\mathbf1_{\mathcal Y\setminus R_Y}
                  \mathbf1_{\mathcal Z\setminus R_Z},\quad
\fun g_r=\fun g'\mathbf1_{\mathcal X\setminus R_X}
                  \mathbf1_{\mathcal Z\setminus R_Z},\quad
\fun h_r=\fun h'\mathbf1_{\mathcal X\setminus R_X}
                  \mathbf1_{\mathcal Y\setminus R_Y}.
\end{align*}
Outside the remainders, the new cells refine the old ones;
on the remainder cells, these functions are zero.
Assign null sets to cells and extend values constantly there.
Thus $\fun c_r:=\fun f_r\fun g_r\fun h_r$ is finite-type.

\paragraph{Step 3. Bound the approximation error.}
For $\fun f_r$,
\[
\E_{Y,Z}|\fun f-\fun f_r|
\le\frac1{6r}+\nu^{\otimes4}(R_Y)+\nu^{\otimes2}(R_Z)
\le\frac1{6r}+\frac{2s}{q}\le\frac1{3r},
\]
and the same bound holds for $\fun g_r,\fun h_r$.
By Claim~\ref{claim:haar_measure_marginal}, both $\mu_0$ and $\mu_1$
have the product Haar marginals used above. Hence, for $i\in\{0,1\}$,
\begin{align*}
|\mu_i(\fun c-\fun c_r)|
&\le\E_{\mu_i}|\fun f\fun g\fun h-\fun f_r\fun g_r\fun h_r|\\
&\le\E_{Y,Z}|\fun f-\fun f_r|
   +\E_{X,Z}|\fun g-\fun g_r|
   +\E_{X,Y}|\fun h-\fun h_r| \\
&\le 1/r.\qedhere
\end{align*}
\end{proof}

\end{document}